%% file: main.tex
\documentclass[12pt]{article}
\usepackage{template}

\title{Overlapping community detection in networks via sparse spectral decomposition}

\author[$1$, $2$]{Jes\'us Arroyo\thanks{jesus.arroyo@jhu.edu}}
\author[$3$]{Elizaveta Levina\thanks{elevina@umich.edu}}
\affil[$1$]{\small Department of  Mathematics, University of Maryland, College Park}
\affil[$2$]{\small Department of Applied Mathematics and Statistics, Johns Hopkins University}
\affil[$3$]{\small Department of Statistics, University of Michigan, Ann Arbor}
\date{}

\begin{document}

\maketitle

\begin{abstract}
We consider the problem of  estimating overlapping community memberships in a network, where each node can belong to multiple communities.  More than a few communities per node are difficult to both estimate and interpret, so we focus on sparse node membership vectors.  Our algorithm is 
 based on sparse principal subspace estimation with iterative thresholding.  The method is computationally efficient, with computational cost equivalent to estimating the leading  eigenvectors of the adjacency
matrix, and does not require an additional clustering step, unlike spectral clustering methods.  We show that a fixed point of the algorithm corresponds to correct node memberships under a version of the stochastic block model. The methods are evaluated empirically on simulated and real-world networks, showing good statistical performance and computational efficiency.
\end{abstract}

\input{intro}
\input{model}

\input{method3}

\input{evaluation}

\input{discussion}

\section*{Acknowledgements}
This research was supported in part by NSF grants DMS-1521551 and DMS-1916222. The authors would like to  thank Yuan Zhang for helpful discussions, and Advanced Research Computing at the University of Michigan for  computational resources and services.

\bigskip

\bibliographystyle{apalike}
\bibliography{biblio}

\input{appendix2}

\end{document}

%% file: intro.tex
\section{Introduction}

Networks have become a popular representation of complex data that appear in different fields such as biology, physics, and the social sciences. A network represents units of a system as nodes, and the interactions between them as edges.   A network can encode relationships between people in a social environment \citep{wasserman1994social}, connectivity between areas of the brain \citep{bullmore2009complex} or interactions between proteins \citep{schlitt2007current}. The constant technological advancements have increased our ability to collect network data on a large scale, with potentially millions of nodes in a network.  Parsimonious models are needed in order to obtain meaningful interpretations of such data, as well as computationally efficient methods.   

Communities are a structure of interest in the analysis of networks, observed in many real-world systems 
\citep{girvan2002community}. Usually, communities are defined as clusters of nodes that have stronger connections to each other than to the rest of the network. Finding these communities allows for a more parsimonious representation of the data which is often meaningful in the system of interest.  For example, communities can represent functional areas of the brain \citep{schwarz2008community,power2011functional}, political affinity in social networks \citep{adamic2005political,conover2011political,latouche2011overlapping}, research areas in citation networks \citep{ji2016coauthorship}, and many others.

The stochastic block model (SBM) \citep{holland1983stochastic} is a simple statistical model for a network with communities, well understood by now; see \cite{Abbe2017} for a review. Under the SBM, a pair of nodes is connected with a probability that only depends on the community memberships of these nodes. The SBM  allows to represent any type of connectivity structure between the communities in the network, such as affinity, disassortativity, or core-periphery (see for example \cite{cape2019spectral}).
While the SBM itself is too simple to capture some aspects of real-world networks, many extensions have been proposed to incorporate more complex structures such as \emph{hubs} \citep{ball2011efficient}, or nodes that belong to more than one community \citep{airoldi2009mixed,latouche2011overlapping,Zhang2020}, which lead to models with overlapping communities. 

Overlapping community models characterize each node by a membership vector, indicating its degree of belonging to different communities.    While in principle all entries of a membership vector can be positive \citep{airoldi2009mixed},
a sparse membership vector is more likely to have  a  meaningful interpretation.  At the same time, allowing for a varying degree of belonging to a community adds both flexibility and interpretability relative to binary membership overlapping community  models such as \cite{latouche2011overlapping}. 
In this paper, we focus on estimating sparse overlapping community membership vectors with continuous entries, so that most nodes belong to only one or few communities, and the degree to which they belong to a community can vary.  The sparsest case where each node belongs to exactly one community corresponds to the classic community detection setting, and its success in modeling and analyzing real-world networks in many different fields \citep{porter2009communities} supports the sparse nature of community memberships.

Existing statistical models for overlapping community detection include both binary membership models, where each node either belongs to a community or does not (e.g., \cite{latouche2011overlapping}), and continuous membership models which allow each node to have a different level of association with each community \citep{airoldi2009mixed,ball2011efficient,psorakis2011overlapping,Zhang2020}.  Binary memberships are a natural way to induce sparsity, but the binary models are less flexible, and fitting them be computationally intensive since they involve solving a discrete optimization problem. On the other hand, continuous memberships are not able to explicitly model sparsity, and the resulting estimates often assign most of the nodes to many or even all communities. To obtain sparse memberships, an ad hoc post-processing step can be applied \citep{gregory2010finding,lancichinetti2011finding}, 
but is likely to lead to a less accurate fit to the data.   Another approach to induce sparse memberships is to incorporate sparsity-inducing priors into the model, for example via the discrete hierarchical Dirichlet process \citep{wang2009decoupling} or the Indian buffet process \citep{williamson2010ibp} that have been introduced in the topic modeling literature.

The problem of estimating overlapping community memberships has been approached from different perspectives; see for example \cite{xie2013overlapping,da2020comparative}, and references in Section \ref{sec:comparison_methods}.  In particular, spectral methods for community detection are popular due to their computational scalability and theoretical guarantees \citep{newman2006finding,rohe2011,lyzinski2014perfect,le2017concentration}. Many statistics network models make a low-rank assumption on the  matrix $\bP=\e[\bA]$  that characterizes the edge probabilities, and in most models with communities the principal subspace of $\bP$ contains the information needed to identify the communities.    Spectral methods for community detection exploit this fact by computing an eigendecomposition of the network adjacency matrix $\bA$, defined by $\bA_{ij} = 1$ if there is an edge from $i$ to $j$ and 0 otherwise, followed by a post-processing step applied to the  leading eigenvectors
    to recover memberships. Several approaches of this type have been recently developed, with different ways of clustering the rows of the leading eigenvectors  \citep{Zhang2020,rubin2017consistency,jin2017estimating,mao2017mixed,mao2018overlapping,mao2020estimating}.

     In contrast to other spectral methods, here we present a new approach for detecting overlapping communities based on estimating a {\em sparse} basis for the principal subspace of the network adjacency matrix in which the pattern of non-zero values  contains the information about community memberships. Our approach can be seen as an analogue to finding sparse principal components of a matrix \citep{jolliffe2003modified,zou2006sparse,ma2013sparse}, but with the important difference that we consider a non-orthogonal sparse basis of the principal subspace to allow for overlaps in communities.   Our method has thus the potential to estimate overlapping  community memberships more accurately than traditional spectral methods, with the same low computational cost of computing the leading eigenvectors of a matrix. We will  demonstrate this both on simulated networks with overlapping and non-overlapping communities, and on real-world networks.

%% file: model.tex
\section{A sparse non-orthogonal eigenbasis decomposition}

As mentioned in the introduction, we consider binary symmetric adjacency matrices $\bA\in\{0,1\}^{n\times n}$, with no self-loops, i.e.,  $\bA_{ii}=0$.  We model the network as an inhomogeneous Erd\"os-R\'enyi random graph \citep{Bollobas2007}, meaning that the upper triangular entries of $\bA$ are independent Bernoulli random variables with potentially different edge probabilities $\bP_{ij}=\p(\bA_{ij}=1)$ for $i,j\in[n]$, $i<j$, contained in a symmetric probability matrix $\bP\in\real^{n\times n}$. 

Our goal is to recover an overlapping community structure in  $\bA$ by estimating an appropriate sparse basis of the invariant subspace of $\bP$.  The rationale is that when $\bP$ is even approximately low rank, most relevant information on communities is contained in the leading eigenvectors of $\bP$, and can be retrieved by looking at a particular basis of its invariant subspace. We will assume that rank of $\bP$ is $K<n$.  The principal subspace of $\bP$ can be described with a full rank matrix $\bV\in\real^{n\times K}$, with columns of $\bV$ forming a basis of this space. Most commonly, $\bV$ is defined as the $K$ leading eigenvectors of $\bP$, but for the purposes of recovering community membership, we focus on finding a sparse \emph{non-negative eigenbasis} of $\bP$, that is, a matrix $\bV$ for which $\bV_{ik}\geq 0$ for all $i\in[n], k\in[K]$ and $\bP = \bV\bU^\top$ for some full rank matrix $\bU\in\real^{n\times K}$.  Note that this is different from the popular non-negative matrix factorization problem \citep{lee1999learning}, as we do not assume that $\bU$ is a non-negative matrix nor do we try to estimate it. 

If $\bP$ has a sparse non-negative basis of its principal subspace $\bV\in\real^{n\times K}$, this basis is not unique, as any column scaling or permutation of $\bV$ will give another non-negative basis.    Among these, we are interested in finding a sparse non-negative eigenbasis $\bV$, since we will relate the non-zeros of $\bV$ to community memberships. The following proposition provides a sufficient condition for identifiability of the non-zero pattern in $\bV$ up to a permutation of its columns.   The proof is included in the Appendix.

\begin{proposition} \label{prop:identifiabilitySPCA}
	Let $\bP\in\Bbb{R}^{n\times n}$ be a symmetric matrix of rank $K$. Suppose that there exist 
	a matrix $\bV\in\Bbb{R}^{n\times K}$ that satisfies the next conditions:
	\begin{itemize}
		\item \emph{Eigenbasis:} $\bV$ is a basis of the column space of $\bP$, that is,  $\bP = \bV\bU^\top$, for some $\bU\in\real^{n\times K}$. 
		\item \emph{Non-negativity:} The entries of $\bV$ satisfy  $\bV_{ik}\geq 0$ for all $i\in[n],k\in[K]$
		\item \emph{Pure rows:} For each $k=1,\ldots, K$ there exists at least one row $i_k$ of $\bV$ such that $\bV_{i_kk}>0$ and $\bV_{i_{k}j}=0$ for $j\neq k$.
	\end{itemize}
	If another matrix $\widetilde{\bV}\in\Bbb{R}^{n\times K}$ satisfies these conditions, then there exists a permutation matrix $\bQ\in\{0,1\}^{K\times K}$, $\bQ^\top \bQ=\bI_K$, such that
	\begin{equation*}
	\supp(\bV)	= \supp(\widetilde{\bV}\bQ),
	\end{equation*}
	where $\supp(\bV) = \{(i,j)|\bV_{ij}\neq 0\}$ is the set of non-zero entries of $\bV$.

\end{proposition}

We connect a non-negative non-orthogonal basis to community memberships through the  \emph{overlapping continuous community assignment model} (OCCAM) of \cite{Zhang2020}, a general model for overlapping communities that encompasses, as special cases, multiple other popular overlapping models \citep{latouche2011overlapping,ball2011efficient,jin2017estimating, mao2018overlapping}.   Under OCCAM, each node is associated with a vector $\bz_i=[z_{i1}, \ldots, z_{iK}]^\top\in\real^K$, $i=1, \ldots,n$, where $K$ is the number of communities in the network. Given $\bZ=[\bz_1 \cdots \bz_n]^\top\in\real^{n\times K}$ and parameters $\alpha>0$, $\bTheta\in\real^{n\times n}$ and $\bB\in\real^{K\times K}$ to be explained below, the probability matrix $\bP=\e[\bA]$ of OCCAM can be expressed as
\begin{equation}
\bP =  \alpha \bTheta \bZ \bB \bZ^T \bTheta. \label{eq:OCCAM}
\end{equation}
For identifiability, OCCAM assumes that $\alpha$ and all entries of $\bTheta$, $\bB$ and $\bZ$ are non-negative, $\bTheta$ is a diagonal matrix with $\text{diag}(\bTheta)=\btheta\in\real^n$ and $\sum_{i=1}^n\btheta_i=n$,  $\|\bz_i\|_2 =\big(\sum_{j=1}^Kz_{ij}^2\big)^{1/2}= 1$ for all $i\in[n]$, and $\bB_{kk}= 1$ for all $k\in[K]$. In this representation, the row $\bz_i$ of $\bZ$ is viewed as the community membership vector of node $i$. A positive value of $z_{ik}$ indicates that node $i$ belongs to community $k$, and the magnitude of $z_{ik}$ determines how strongly. The parameter $\btheta_i$ represents the degree correction for node $i$ as in the degree-corrected SBM \citep{karrer2011stochastic}
allowing for degree heterogeneity and in particular ``hub" nodes, common in real-world networks. The scalar parameter  $\alpha>0$  controls the edge density of the entire graph.

One can obtain the classical SBM as a special case of OCCAM by further requiring each $\bz_i$ to have only one non-zero value and setting $\btheta_i=1$ for all $i\in[n]$.
Keeping only one non-zero value in each row of $\bZ$ but allowing the entries of $\btheta$ to take positive values, one can recover the degree-corrected SBM  \citep{karrer2011stochastic}.  
More generally in OCCAM, nodes can belong to multiple communities at the same time. Each row of $\bZ$  can have multiple or all the entries different from zero, indicating the communities to which the node belong.

Equation~\eqref{eq:OCCAM} implies that under OCCAM the probability matrix $\bP$  has a non-negative eigenbasis given by $\bV = \bTheta\bZ$. The following proposition shows the converse result, namely, that any matrix $\bP$ that admits a non-negative eigenbasis can be represented as in Equation~\eqref{eq:OCCAM}, which motivates the interpretation of the non-zero entries of a non-negative eigenbasis as indicators of  community memberships.

\begin{proposition} \label{prop:V-occam}
	Let $\bP\in\real^{n\times n}$ be a symmetric real matrix with $\rank(\bP)=K$. Suppose that there exists a full-rank nonnegative matrix $\bV\in\real^{n\times K}$ and a matrix $\bU$ such that $\bP = \bV\bU^\top$. Then, there exists a non-negative diagonal matrix $\bTheta\in\real^{n\times n}$, a non-negative matrix $\bZ\in\real^{n\times K}$ with $\sum_{k=1}^K\bZ_{ik}^2=1$ for each $i\in[n]$, and a symmetric matrix $\bB\in\real^{K\times K}$ such that
	$$\bP=\bTheta\bZ\bB\bZ^\top \bTheta.$$
	Moreover, if $\bV$ satisfies the conditions of Proposition \ref{prop:identifiabilitySPCA}, then $\supp(\bV)=\supp(\bZ\bQ)$ for some permutation $\bQ\in\real^{K\times K}, \bQ^\top\bQ=\bI$.
\end{proposition}

In short, Proposition \ref{prop:V-occam} states a non-negative basis of the probability matrix $\bP$ can be mapped to overlapping communities as in Equation~\eqref{eq:OCCAM}. Moreover, under the conditions on this eigenbasis stated in Proposition~\ref{prop:identifiabilitySPCA}, the community memberships can be uniquely identified. These conditions are weaker than the ones in \cite{Zhang2020} since we are only interested in community memberships and not in identifiability of the other parameters; note that we do not aim to fit the OCCAM model, which is computationally much more intensive than our approach here. Other conditions for identifiability of overlapping community memberships have been presented in the literature \citep{huang2019detecting}, but the pure row assumption in Proposition~\ref{prop:identifiabilitySPCA} is enough for our purpose of estimating sparse memberships.

%% file: method3.tex
\section{Community detection via sparse iterative thresholding}

Our goal is to compute an appropriate sparse basis of the principal subspace of $\bA$ which contains  information about the overlapping community  memberships. Spectral clustering has been popular for community detection, typically clustering the rows of the leading eigenvectors of $\bA$ or a function of them to assign nodes to communities.  Spectral clustering with overlapping communities typically gives a continuous membership matrix, which can then be thresholded to obtain sparse membership vectors; however, this two-stage approach is unlikely to be optimal in any sense, and some of the overlapping clustering procedures can be computationally expensive \citep{Zhang2020,jin2017estimating}. In contrast, our approach of directly computing a sparse basis of the principal subspace of $\bA$ avoids the two-stage procedure and thus can lead to improvements in both accuracy and computational efficiency.

Sparse principal component analysis (SPCA)  \citep{jolliffe2003modified,zou2006sparse} seeks to estimate  the principal subspace of a matrix while incorporating sparsity constraints or regularization on the basis vectors. In high dimensions, enforcing sparsity can improve estimation when the sample size is relatively small, and/or simplify the interpretation of the solutions.  Many SPCA algorithms have been proposed to estimate eigenvectors of a matrix under sparsity assumptions (see for example \cite{amini2008high,johnstone2009consistency,vu2013minimax,ma2013sparse}).

Our goal is clearly related to SPCA since we are interested in estimating a sparse basis of the principal subspace of $\bP$, but an important difference is that our vectors of interest are not necessarily  orthogonal; in fact orthogonality is only achieved when estimated communities do not overlap, and is thus not compatible with meaningful overlapping community estimation.  
For non-overlapping community detection, however, there is a close connection between a convex relaxation of the maximum likelihood estimator of communities and a convex formulation of SPCA \citep{amini2018semidefinite}.

 Orthogonal iteration is a classical method for estimating the eigenvectors of a matrix; see for example \cite{golub2012matrix}. \cite{ma2013sparse} extended this method  to estimate sparse eigenvectors by an iterative thresholding algorithm.  Starting from an initial matrix $\bV^{(0)}\in\mathbb{R}^{n\times K}$, the general form of their algorithm iterates the following steps until convergence:
\begin{enumerate}
	\item Multiplication step:
	\begin{equation}
	\bT^{(t+1)} = \bA \bV^{(t)}. \label{eq:general-spca-multiplication}
	\end{equation}	
	\item Regularization step:
	\begin{equation}
	\bU^{(t+1)} = \mathcal{R}(\bT^{(t)}, \bLambda), \label{eq:general-spca-regularization}
	\end{equation}		
	where $\mathcal{R}:\mathbb{R}^{n\times K}\rightarrow\mathbb{R}^{n\times K}$ is a regularization function and $\bLambda\in\real^{n\times K}$ a matrix of regularization parameters.
	\item Identifiability  step:	
	\begin{equation}
	\bV^{(t+1)} = \bU^{(t+1)}\bW^{(t+1)}, \label{eq:general-spca-numericalstab}
	\end{equation}
	where $\bW^{(t+1)}$ is a $K\times K$ matrix.
\end{enumerate}
An example of a convergence criterion may be stopping when the distance between subspaces generated by $\bV^{(t)}$ and $\bV^{(t+1)}$ is small. For two full-rank matrices $\bU,\tilde{\bU}\in\Bbb{R}^{n\times K}$, the distance between the subspaces generated by the columns of $\bU$ and $\tilde{\bU}$ is defined through their orthogonal projection matrices $\bR=\bU(\bU^\top\bU)^{-1}\bU^\top$ and $\tilde{\bR} = \tilde{\bU}(\tilde{\bU}^\top\tilde{\bU})^{-1}\tilde{\bU}^\top$ as
\begin{equation*}
\operatorname{dist}(\bU, \tilde{\bU}) = \|\bR - \tilde{\bR}
\|,
\end{equation*}
where $\|\cdot \|$ is the matrix spectral norm   (see \cite{golub2012matrix}, Section 2.5.3).

Let $\widehat{\bV}$ be the value of $\bV^{(t)}$ at convergence, and let $\widetilde{\bV}$ be the $n\times K$ matrix of the $K$ leading eigenvectors of $\bA$. The algorithm provides a generic framework for obtaining a basis $\widehat{\bV}$ that is close to $\widetilde{\bV}$, and the regularization step can be customized to enforce some structure in $\widehat{\bV}$.   In each iteration, the multiplication step \eqref{eq:general-spca-multiplication} reduces the distance between the subspaces generated by $\bV^{(t)}$ and  $\widetilde{\bV}$ (Theorem 7.3.1 of \cite{golub2012matrix}), and then the regularization step~\eqref{eq:general-spca-regularization} forces some structure in $\bV^{(t)}$. \cite{ma2013sparse} focused on sparsity and regularized with a thresholding function satisfying $|[\mathcal{R}(\bT,\bLambda)]_{ik} - \bT_{ik}|\leq \bLambda_{ik}$ and $[\mathcal{R}(\bT,\bLambda)]_{ik}\mathbbm{1}(|\bT_{ik}|\leq \bLambda_{ik})=0$ for all $\bLambda_{ik}>0$ and $i\in[n],k\in[K]$, which includes both hard and soft thresholding. If the  distance between $\bU^{(t)}$ and $\bV^{(t)}$ is small, then the distance between $\bV^{(t)}$ and $\tilde{\bV}$ keeps decreasing until a certain tolerance is reached (Proposition 6.1 of \cite{ma2013sparse}).  Finally, the last step in Equation~\eqref{eq:general-spca-numericalstab} ensure identifiability. For example, the orthogonal iteration algorithm uses the QR decomposition $\bQ^{(t)}\bR^{(t)}=\bU^{(t)}$ and sets $\bV^{(t)}=\bU^{(t)}\bR^{(t)^{-1}}$, which is an orthogonal matrix.

We will use the general form of the algorithm presented in  Equations~\eqref{eq:general-spca-multiplication}-\eqref{eq:general-spca-numericalstab} to develop methods for estimating a sparse eigenbasis of $\bA$, by designing regularization and identifiability steps appropriate for overlapping community detection.

\subsection{Sparse eigenbasis  estimation}

We propose an iterative thresholding algorithm for sparse eigenbasis estimation when the basis vectors are not necessarily orthogonal. Let $\bV^{(t)}$ be the estimated basis at iteration $t$. For identifiability, we assume that this matrix has  normalized columns, that is, $\|\bV^{(t)}_{\cdot,k}\|_2 =1$ for each $k\in[K]$, where $\bV_{\cdot,k}$ denotes  the $k$-th column of $\bV$. Our algorithm is based on the following heuristic. Suppose that at some iteration $t$, $\bV^{(t)}$ is close to the basis of interest. The multiplication step in Equation~\eqref{eq:general-spca-multiplication} moves $\bV^{(t)}$ closer to  $\widetilde{\bV}$, the $K$-leading eigenspace of $\bA$, but the entries of $\bT^{(t+1)}=\bA\bV^{(t)}$ and $\bV^{(t)}$ are not necessarily close. Hence, before applying the regularization step, we introduce a linear transformation step that returns $\bT^{(t+1)}$ to a value that is close to $\bV^{(t)}$ entry-wise. This transformation is given by the solution of the  optimization problem
$$ \bGamma^{(t+1)} = \argmin_{\bGamma\in\real^{K\times K}} \|\bV^{(t)}\bGamma - \bT^{(t+1)}\|_F^2,$$ 
which has a closed form solution, $\bGamma^{(t+1)} = \left[ (\bV^{(t)})^\top \bV^{(t)}\right]^{-1}(\bV^{(t)})^\top\bT^{(t+1)}$. Define $$\widetilde{\bT}^{(t+1)} = \bT^{(t+1)} \left(\bGamma^{(t+1)}\right)^{-1}. $$
After this linear transformation, we apply a sparse regularization to $\widetilde{\bT}^{(t+1)}$,  defined by a thresholding function $\mathcal{S}$ with parameter $\lambda\in[0,1)$, 
\begin{equation}
\left[\mathcal{S}(\widetilde{\bT}, \lambda)\right]_{ik} = \left\{\begin{array}{cl}
\widetilde{\bT}_{ik} & \text{if }\widetilde{\bT}_{ik}> \lambda \max_{j}|\widetilde{\bT}_{ij}|,\\
0 & \text{otherwise}.
\end{array} \right. \label{eq:thresholdingfunction}
\end{equation}
The function $\mathcal{S}$ applies hard thresholding to each entry of the matrix $\widetilde{\bT}$ with a different threshold for each row, to adjust for possible differences in the expected degree of a node. 
The parameter $\lambda$ controls the level of sparsity, with larger values of $\lambda$ giving more zeros in the solution.  Finally, the new value of $\bV$ is obtained by normalizing the columns, setting $\bU^{(t+1)} = \mathcal{S}(\widetilde{\bT}^{(t+1)}, \lambda)$ and 
$$\bV^{(t+1)}_{i k} = \frac{1}{\|\bU^{(t+1)}_{\cdot, k}\|_2} \bU^{(t+1)}_{ik},$$
for each $i\in[n]$ and $k\in[K]$.

We stop the algorithm after the relative difference in spectral norm between $\bV^{(t)}$ and $\bV^{(t+1)}$ is smaller than some tolerance $\epsilon>0$, that is, 
\begin{equation*}
\frac{\|\bV^{(t+1)} - \bV^{(t)}\|}{\|\bV^{(t)}\|} < \epsilon.
\end{equation*}
These steps are summarized in Algorithm \ref{alg:spca-eig}.

\begin{algorithm}
	\caption{SPCA-eig: Sparse Eigenbasis Estimation}
	\begin{algorithmic}
		\Input Adjacency matrix $\bA$, regularization parameter $\lambda\in[0,1)$, initial estimator $\bV^{(0)}$ with normalized columns.
		\For{$t=1,\ldots$ until convergence}
		\begin{enumerate}
		    \item Update $\bT^{(t)} =\bA \bV^{(t-1)}$.
		    \item Update $\widetilde{\bT}^{(t)} = \bT^{(t)}[(\bV^{(t-1)})^\top \bT^{(t)}]^{-1}[ (\bV^{(t-1)})^\top\bV^{(t-1)}]$.
		    \item Thresholding: $\bU^{(t)} = \mathcal{S}(\widetilde{\bT}^{(t)}, \lambda)$.
		    \item Normalization: $\bV_{\cdot, k}^{(t)} = 
		    \frac{1}{\|\bU^{(t)}_{\cdot,k}\|_2}\bU^{(t)}_{\cdot,k}$, for $k\in[K]$.
		\end{enumerate}
		\EndFor\\
		\Return $\hat{\bV}_\lambda=\bV^{(t)}$ the value at convergence.
	\end{algorithmic}
	\label{alg:spca-eig}
\end{algorithm}

The following proposition shows that when Algorithm~\ref{alg:spca-eig} is applied to the expected probability matrix $\bP$ that has a sparse basis, then there exists a fixed point that has the correct support. In particular, this implies that for the expected probability matrix of an OCCAM graph defined in Equation~\eqref{eq:OCCAM},
the entries of this fixed point coincide with the support of the overlapping memberships of the model. The proof is given on the Appendix.
\begin{proposition} \label{prop:fixed-point1}
Let $\bP\in\real^{n\times n}$ be a symmetric matrix
with $\rank(\bP)=K<n$, and suppose that there exists a non-negative sparse basis $\bV$ of the principal subspace of $\bP$.
Let $\widetilde{\bV}$ be a matrix such that $\widetilde{\bV}_{\cdot,k}=\frac{1}{\|\bV_{\cdot,k}\|_2}\bV_{\cdot,k}$ for each $k\in[K]$, and  $v^\ast= \min\{\bV_{ik}/\|\bV_{i\cdot}\|_\infty|(i,k)\in\supp(\bV)\}$. Then, for any $\lambda\in[0,v^\ast)$, the matrix $\widetilde{\bV}$ is a fixed point of Algorithm~\ref{alg:spca-eig} applied to $\bP$.

\end{proposition}

When the algorithm is applied to an adjacency matrix $\bA$, the matrix $\bV$ is not exactly a fixed point, but the norm of the difference between $\bV$ and $\bV^{(1)}$ will be a function of the distance between the principal subspaces of $\bA$ and $\bP$. Concentration results \citep{le2017concentration} ensure that $\bA$ is close to its expected value $\bP$, specifically, 
 $\|\bA-\bP\|=O(\sqrt{d})$ (where $\|\cdot\|$ is the spectral norm of a matrix) with high probability as long as the largest expected degree $d=\max_{i\in[n]}\bP_{ij}$ satisfies $d=\Omega(\log n)$. If the $K$ leading eigenvalues of $\bV$ are sufficiently large, then the principal subspaces of $\bA$ and $\bP$  are close to each other \citep{yu2015useful},

\subsubsection{Community detection in networks with homogeneous degrees}

Here, we present a second algorithm for the estimation of sparse community memberships in graphs for homogeneous expected degree of nodes within a community. Specifically, we focus on graphs for which the expected adjacency matrix $\bP=\mathbb{E}[\bA]$ has the form
\begin{equation}
    \bP = \bZ\bB\bZ^\top, \label{eq:occam-notheta}
\end{equation}
where $\bZ\in\real^{n\times K}$ is a membership matrix such that $\|\bZ_{i,\cdot}\|_1 = \sum_{k=1}^K|\bZ_{ik}|=1$, and $\bB\in\real^{K\times K}$ is a full-rank matrix. This model is a special case of OCCAM, when the degree heterogeneity parameter $\bTheta$ in  Equation~\eqref{eq:OCCAM} is constant for all vertices. In particular, this case includes the classic SBM \citep{holland1983stochastic} when the memberships do not overlap.

To enforce degree homogeneity, we add an additional normalization step, so that the matrix $\hat{\bZ}$ has rows with constant norm $\|\hat{\bZ}_{i,\cdot}\|_1=1$ as in Equation~\eqref{eq:occam-notheta}. In practice we observed that this normalization gives very accurate results in terms of community detection. After the multiplication step $\bT^{(t)} =\bA\bV^{(t-1)}$, the columns of $\bT^{(t)}$ are proportional to the norm of the columns $\bV^{(t-1)}$,  which is in turn proportional to the estimated community sizes.   In order to remove the effect of this scaling with community size, which is not meaningful for community detection, we normalize the columns of $\bV^{(t-1)}$, and then perform the thresholding and the row normalization step as before.   These steps are summarized  in Algorithm \ref{alg:spca-comdet}.

\begin{algorithm}
	\caption{SPCA-CD: Sparse eigenbasis estimation for networks with homogeneous degrees.} 
	\begin{algorithmic}
		\Input Adjacency matrix $\bA$, regularization parameter $\lambda\in[0,1)$, initial estimator $\bV^{(0)}$ with normalized rows.
		\For{$t=1,\ldots$ until convergence}
		    \begin{enumerate}
		        \item Update $\bT^{(t)} =\bA \bV^{(t-1)}$.
		        \item Column normalization: $\widetilde{\bT}_{\cdot, k}^{(t)} = 
		    \frac{1}{\|\bT^{(t)}_{\cdot,k}\|_1}\bT^{(t)}_{\cdot,j}$, for $k\in[K]$.
		    \item Thresholding: $\bU^{(t)} = \mathcal{S}(\widetilde{\bT}^{(t)}, \lambda)$.
		    \item Row normalization: $\bV^{(t)}_{i,\cdot} = \frac{1}{\|\bU^{(t)}_{i,\cdot}\|_1} \bU^{(t)}_{i,\cdot}$, for $i\in[n]$.
		    \end{enumerate}
		\EndFor\\
		\Return $\hat{\bV}_\lambda=\bV^{(t)}$ the value at convergence.
	\end{algorithmic}
	\label{alg:spca-comdet}
\end{algorithm}

The next theorem shows that in the case of the planted partition SBM, a matrix with the correct sparsity pattern is a fixed point of Algorithm \ref{alg:spca-comdet}. Note that since the algorithm does not assume that each node belongs to a single community, this result not only guarantees that there exist a fixed point that correctly cluster the nodes into communities, as typical goal of of community detection, but also that is able to distinguish if a node belongs to more than one community or not. The proof is given on the appendix.

\begin{theorem}\label{theo-sparsistency}
	Let $\bA$ be a network generated from a SBM with $K$ communities of sizes $n_1,\ldots,n_K$, membership matrix $\bZ$  and connectivity matrix $\bB\in[0,1]^{K\times K}$ of the form
	\begin{equation*}
	\bB_{rs} = \left\{\begin{array}{cl}
	p, & \text{ if $r=s$},\\
	q, & \text{if $r\neq s$},
	\end{array}\right.
	\end{equation*}
	for some $p,q\in[0,1]$, $p>q$. Suppose that for some $\lambda^\ast\in (0,1)$ and some $c_1>2$,
	\begin{equation}
	\lambda^\ast p- q>c_1\sqrt{\frac{\log(Kn)}{\min_{k} n_k}}, \label{theo-condseparation}
	\end{equation}
	Then,  for any $\lambda\in(\lambda^\ast,1)$,  $\bZ$ is a fixed point of Algorithm \ref{alg:spca-comdet} with probability at least $1-n^{c_1-1}$.
\end{theorem}

\subsection{Selecting the thresholding parameter \label{sec:tuningparameters}}

Our algorithms require two user-supplied parameters: the number of communities $K$ and the threshold level $\lambda$.    The parameter $\lambda$ controls the sparsity of the estimated basis $\hat{\bV}$. In practice, looking at the full path of solutions for different values of $\lambda$ may be informative, as controlling the number of overlapping memberships  can result in different community assignments. On the other hand, it is practically useful to select an appropriate value $\lambda$ that provides a good fit to the data. We discuss two possible techniques for choosing this parameter, the Bayesian Information Criterion (BIC) and edge cross-validation (ECV) \citep{li2016network}. Here we assume that the number of communities $K$ is given, but choosing the number of communities is also an important problem, with multiple methods available  for solving it \citep{wang2015likelihood,le2015estimating,li2016network}. If computational resources allow, $K$ can be chosen by  cross-validation along with $\lambda$.

The goodness-of-fit can be measured via the likelihood of the model for the graph $\bA$, which depends on the probability matrix $\bP=\Bbb{E}[\bA]$. Given $\widehat{\bV}$, a natural estimator for $\bP$ is the projection of $\bA$ onto the subspace spanned by $\widehat{\bV}$, which can be formulated as

\begin{align}
\label{eq:projection-into-V}
\widehat{\bP} =  \underset{\bP}{\argmin} &
 \|\bA - \bP\|_F^2 \nonumber \\
 \text{subject to } &  \bP= \widehat{\bV}\bB \widehat{\bV}^\top, \; \bB\in\real^{K\times K}.
\end{align}
This optimization problem finds the least squares estimator of a matrix constrained to the set of symmetric matrices with a principal subspace defined by $\widehat{\bV}$, and has a closed-form solution, stated in the following proposition.   
\begin{proposition}\label{prop:estimator}
Let $\widehat{\bP}$ be the solution of the optimization problem~\eqref{eq:projection-into-V}, and suppose that $\widehat{Q}\in\real^{n\times K}$ is a matrix with orthonormal columns such that $\widehat{\bV}=\widehat{\bQ}\widehat{\bR}$ for some matrix $\widehat{\bR}\in\real^{K\times K}$. Then, 
$$\widehat{\bP} = \widehat{\bQ}\left(\widehat{\bQ}^\top\bA\widehat{\bQ}\right)\widehat{\bQ}^\top.$$
\end{proposition}

The Bayesian Information Criterion (BIC) \citep{schwarz1978estimating} provides a general way of choosing a tuning parameter by balancing the fit of the model measured with the log-likelihood of $\bA$, and a penalty for the complexity of a model that is proportional to the number of parameters. The number of non-zeros in $\bV$ given by $\|\bV \|_0$ can be used as a proxy for the degrees of freedom, and the sample size is taken to be the number of independent edges in $\bA$.   Then the BIC for a given $\lambda$ can be written as
\begin{equation}
\label{eq:BIC_occam}
\mathcal{P}_{\operatorname{BIC}}(\lambda) =  -2\ell(\widehat{\bP}_\lambda) + \|\widehat{\bV}_\lambda\|_0\log(n(n-1)/2), 
\end{equation}
where $\widehat{\bP}_\lambda$ is the estimate for $\bP$ defined in Proposition~\ref{prop:estimator} for $\widehat{\bV}_\lambda$.

The BIC criterion~\eqref{eq:BIC_occam} has the advantage of being simple to calculate, but it has some issues. First, the BIC is derived for a maximum likelihood estimator, while $\widehat{\bP}$ is not obtained in this way, and this is only a heuristic.   Further, the least squares estimator $\widehat{\bP}$ is not guaranteed to result in a valid estimated edge probability (between 0 and 1). A possible practical solution is to modify the estimate by defining $\widetilde{\bP}\in[0,1]^{n\times n}$ as $\widetilde{\bP}_{ij} = \min(\max(\widehat{\bP}_{ij}, \epsilon), 1-\epsilon)$ for some small value of $\epsilon\in(0,1)$.

Another alternative for choosing the tuning parameter is edge cross-validation (CV). \cite{li2016network} introduced  a CV method for network data based on splitting the set of node pairs $\mathcal{N}=\{(i,j):i,j\in\{1,\ldots,N\}\}$ into $L$ folds. For each fold $l = 1, \dots, L$, the corresponding set of node pairs $\Omega_l\subset\mathcal{N}$ is excluded, and the rest are used to fit the basis $\bV$.
\cite{li2016network} propose to use a matrix completion  algorithm based on the rank $K$ truncated SVD to fill in the entries missing after excluding  $\Omega_l$, resulting in a matrix $\hat{M}_l\in\real^{n\times n}$. Then, for a given $\lambda$ we estimate  $\widehat{\bV}_\lambda$, and use Proposition~\ref{prop:estimator} to obtain an estimate $\widehat{\bP}_\lambda(\hat{M}^{(l)})$ of $\bP$.  The error on the held-out edge set is measured by 
\begin{equation}
\mathcal{P}_{\operatorname{CV}}(\bA, \widehat{\bP}_\lambda(\widehat{\bM}_l); \Omega_l) = \frac{1}{|\Omega_l|} \sum_{(i,j)\in\Omega_l} (\bA_{ij} - \hat{\bP}_\lambda(\hat{\bM}_l)_{ij} )^2, 
\end{equation}
and the tuning parameter $\lambda$ is selected to minimize the average cross-validation error $$\mathcal{P}_{\operatorname{CV}}(\lambda)=\frac{1}{L}\sum_{l=1}^L\mathcal{P}_{\operatorname{CV}}(\bA, \widehat{\bP}_\lambda(\widehat{\bM}_l); \Omega_l).$$

The edge CV method does not rely in a specific model for the graph, which can be convenient in the settings mentioned before, but its computational cost is larger. In practice, we observe that edge CV tends to select more complex models in which nodes are assigned to more communities than the solution selected with BIC (see Section~\ref{sec:tuningCVBIC}).

%% file: evaluation.tex
\section{Numerical evaluation on synthetic networks}
We start with evaluating our methods and comparing them to benchmarks on simulated networks.  In all scenarios, we generate networks from OCCAM, thus edges of $\bA$ are independent Bernoulli  random variables, with expectation given by~ \eqref{eq:OCCAM}. We assume that each row vector $\bz_i\in\real^K$ of $\bZ=[\bz_1, \ldots, \bz_n]^\top$ satisfies $\|\bz_i\|_1=1$, so each node has the same expected degree.   To better understand what affects the performance, we evaluate the methods by varying one parameter from the following list at a time;  all of them affect the difficulty of detecting overlapping communities.  
\begin{enumerate}[label=\alph*)]
	\item Fraction of nodes belonging to more than one community $\widetilde{p}$ (the higher $\widetilde p$, the more difficult the problem).  For a given $\widetilde p\in[0,1)$, we select $\widetilde{p}n$ nodes for the overlaps, and assign the rest to only one community, distributed equally among all the communities. For most of the experiments we use $K=3$ communities, and $1/4$ of the overlapping nodes are assigned to all communities with $\bz_i=[1/3, 1/3, 1/3]^T$, while the rest are assigned to two communities $j,k$, with $\bz_{ij}=\bz_{ik}=1/2$, equally distributing these nodes on all pairs $(j,k)$. When $K>3$, we only assign the nodes to one or two communities following the same process, but we do not include overlaps with more than three communities.
	\item Connectivity between communities $\rho$ (the higher $\rho$, the more difficult the problem).  We parameterize $\bB$ as
	\[\bB = (1-\rho) \bI_{K} + \rho \mathbf{1}_K\mathbf{1}_K^\top,\]
	and vary $\rho$ in a range of values between 0 and 1.
	\item Average degree of the network $d$ (the higher $d$, the easier the problem). For a given average degree $d$, we set $\alpha$ in~\eqref{eq:OCCAM} so that the expected average degree $\frac{1}{n}\mathbf{1}^\top_n\Bbb{E}[\bA] \mathbf{1}_n$ is equal to $d$.
	\item Node degree heterogeneity (the more heterogeneous the degrees, the harder the problem).   This is controlled by parameter $\btheta=\text{diag}(\bTheta)$, and in most simulations we set $\btheta_i=1 \ \forall i\in[n]$ so all nodes have the same degree, but in some scenarios we also introduce hub nodes by setting $\theta_i=5$ with probability $0.1$ and $\theta_i=1$ with probability 0.9.
	\item Number of communities $K$ (the larger $K$, the harder the problem).  For all values of $K$, we maintain communities of equal size. 
\end{enumerate}
In most scenarios, we fix $n= 500$, and $K=3$. All simulation settings are run 50 times, and the average result together with its 95\% confidence band are reported. An implementation of the method in R can be found at \url{https://github.com/jesusdaniel/spcaCD}.

Our main goal is to find the set of non-zero elements of the membership matrix. Many measures can be adopted to evaluate a solution;  here we use the \emph{normalized variation of information} (NVI) introduced by
\cite{lancichinetti2009detecting}, which is specifically designed for problems with overlapping clusters. Given a pair of binary random vectors $X, Y$ of length $K$, the normalized conditional entropy of $X$ with respect to $Y$ can be defined as
\begin{equation*}
H_{norm}(X|Y) = \frac{1}{K}\sum_{k=1}^K\frac{H(X_k|Y_k)}{H(X_k)},
\end{equation*}
where $H(X_k)$ is the entropy of $X_k$ and $H(X_k|Y_k)$ is the conditional entropy of $X_k$ given $Y_k$, defined as
\begin{eqnarray}
H(X_k) & = & -P(X_k=0)\log P(X_k=0) - P(X_k=1)\log P(X_k=1) \label{eq:entropyX}\\
H(X_k, Y_k) & = & -\sum_{a=0}^1\sum_{b=0}^1P(X_k=a, Y_k=b)\log P(X_k=a, Y_k=b) \label{eq:entropyXY}\\
H(X_k|Y_k) & = & H(X_k, Y_k) - H(Y_k),\nonumber
\end{eqnarray}
and the normalized variation of information between $X$ and $Y$ is defined as
\begin{equation}
N(X|Y) = 1- \min_\sigma\frac{1}{2}\left(H_{norm}(\sigma(X)|Y) + H_{norm}(Y|\sigma(X))\right),
\end{equation}
where $\sigma$ is a permutation of the indexes to account for the fact that the binary assignments can be equivalent up to a permutation.   The NVI is always a number between 0 and 1; it is equal to 0 when $X$ and $Y$ are independent, and to 1 if $X=Y$.

  For a given pair of binary membership matrices $\bZ$ and $\widetilde{\bZ}$ with binary entries indicating community memberships, we can use the rows of  replace the probabilities in equations \eqref{eq:entropyX} and \eqref{eq:entropyXY} with the sample versions using the rows of $\widetilde{\bZ}$ and $\bZ$, that is
$$\hat{P}(X_k=a) = \frac{1}{n}\sum_{i=1}^n\mathrm{1}\{ \widetilde\bZ_{ik}=a \},\quad\quad\quad\hat{P}(Y_k=b) = \frac{1}{n}\sum_{i=1}^n\mathrm{1}\{\bZ_{ik}=b\},$$
$$\hat{P}(X_k=a, Y_k=b) = \frac{1}{n}\sum_{i=1}^n\mathrm{1}\{\widetilde{\bZ}_{ik}=a, \bZ_{ik}=b\},$$
for $a,b\in\{0,1\}.$

\subsection{Choice of initial value}
We start from comparing several initialization strategies: 
\begin{itemize}
	\item An overlapping community assignment, from the method for fitting OCCAM. 
	\item A non-overlapping community assignment, from SCORE \citep{jin2015}, a spectral clustering method designed for networks with heterogeneous degrees.
	\item  Multiple random non-overlapping community assignments, with each node randomly assigned to only one community.   We use five different random values and take the solution corresponding the smallest error as measured by~\eqref{eq:projection-into-V}.
\end{itemize}

We compare these initialization schemes with fixed $n=500$, $K=3$, $d=50$, and varying between-community connectivity $\rho$ and the fraction of overlapping nodes $\widetilde p$. For both our methods (SPCA-eig and SPCA-CD), we fit solution paths over a range of values $\lambda=\{0.05, 0.1, \ldots,0.95\}$, and report the solution with the highest NVI for each of the methods (note that we are not selecting $\lambda$ in a data-driven way in order to reduce variation that is not related to initialization choices).

Figure \ref{fig:initializationSPCA} shows the results on initialization strategies. In general, all methods perform worse as the problem becomes harder, and the non-random initializations perform better overall;  the multiple random initializations are also sufficient for the easier case of few nodes in overlaps.  
For the rest of the paper, unless explicitly stated, we use the non-overlapping community detection solution (SCORE) to initialize the algorithm, given its good performance and low computational cost.

\begin{figure}
	\centering
	\includegraphics[width=\textwidth]{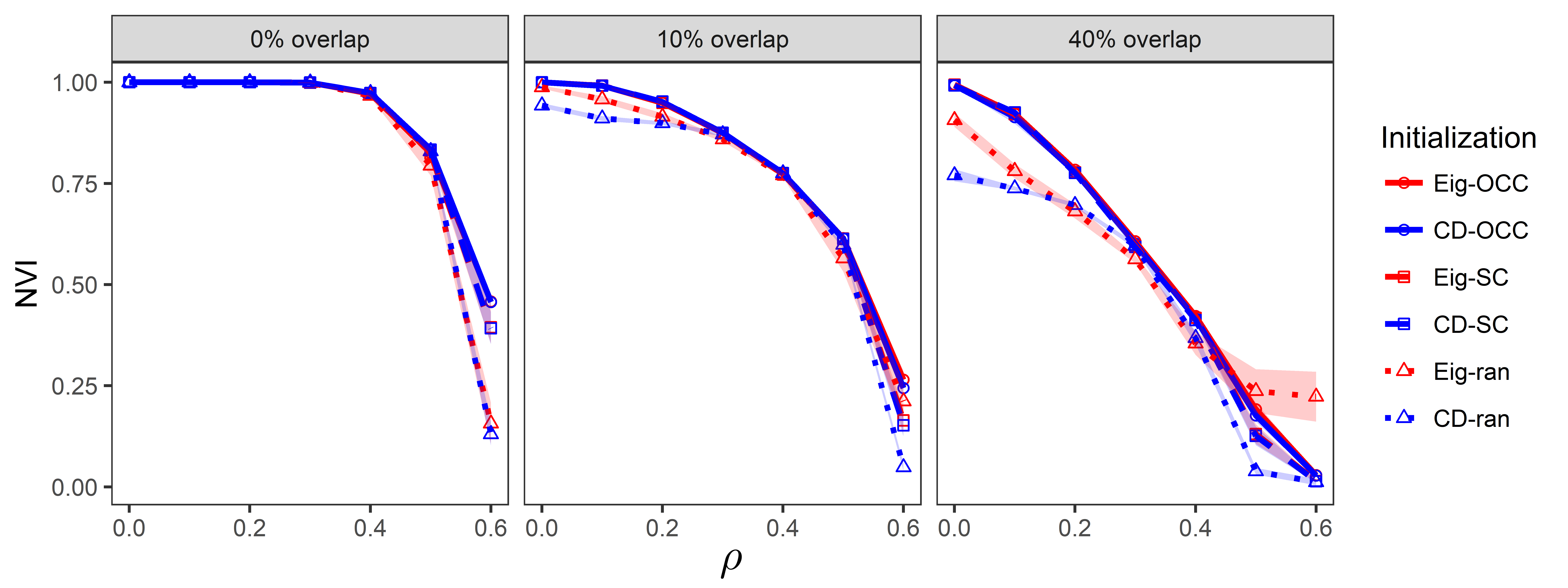}
	\caption{Performance measured by NVI under different initialization strategies (OCCAM, SCORE and five random initial values). The errors are plotted as a function of  connectivity between communities $\rho$ for three different values of the overlap $\widetilde p$.}
	\label{fig:initializationSPCA}
\end{figure}

\subsection{Choosing the threshold \label{sec:tuningCVBIC}}
The tuning parameter $\lambda$ controls the sparsity of the solution, and hence the fraction of pure nodes.   Since community detection is an unsupervised problem, it may be useful in practice to look at the entire path over $\lambda$ and consider multiple solutions with different levels of sparsity (see Section \ref{sec:karate}). However, we may also want to choose a single value of $\lambda$ that balances a good fit and a parsimonious solution. Here, we evaluate the performance of the two strategies for choosing $\lambda$ proposed in Section \ref{sec:tuningparameters}, BIC and CV, using the same simulation setting than the previous section.

Figure \ref{fig:tuningSPCA} shows the average performance measured by NVI of the two tuning methods.   The BIC tends to select sparser solutions than CV, and hence when the true membership matrix is sparse (few overlaps), BIC outperforms CV, but with more overlap in communities, CV usually performs better, specially for SPCA-CD. Since there is no clear winner overall, we use BIC in subsequent analysis, because it is computationally cheaper.

\begin{figure}
	\centering
	\includegraphics[width=\textwidth]{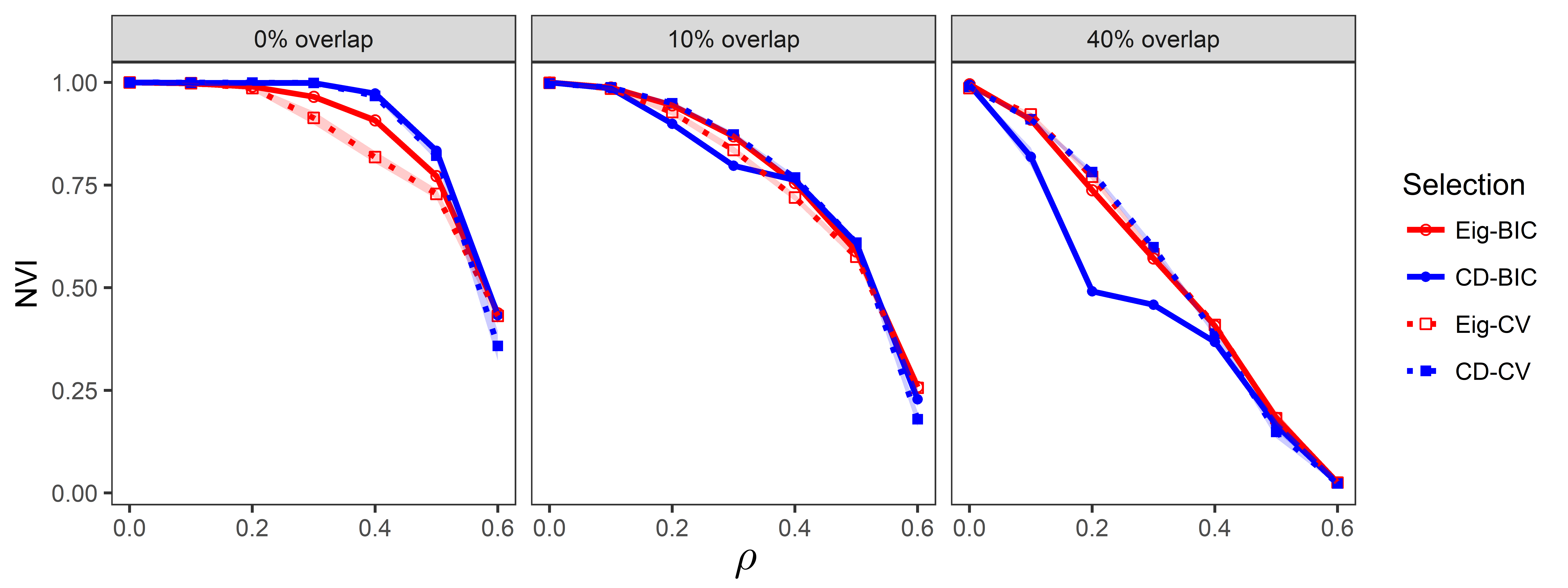}
	\caption{Performance of two tuning methods, BIC and CV, as measured by NVI.   The errors are plotted as a function of  connectivity between communities $\rho$ for three different values of the overlap $\widetilde p$. 
	}\label{fig:tuningSPCA}
\end{figure}

\subsection{Comparison with existing methods \label{sec:comparison_methods}}
We compare our proposal to several state-of-the-art methods for overlapping community detection. We use the same simulation settings as in the previous section ($n=500$ and $K=3$), including sparser scenarios with $d=20$, and networks with heterogeneous degrees ($d=50$ and 10\% of nodes are hubs). 

We select competitors based on good performance reported in previous studies.    As representative examples of spectral methods, we include OCCAM fitted by the algorithm  of \cite{Zhang2020} and Mixed-SCORE \citep{jin2017estimating}.   We also include the EM algorithm for the BKN model \citep{ball2011efficient} and the overlapping SBM of \cite{latouche2011overlapping} (OSBM), and Bayesian non-negative matrix factorization (BNMF) by \cite{psorakis2011overlapping}.  
For methods that return a continuous membership assignment (OCCAM, BKN and Mixed-SCORE), we follow the approach of \cite{Zhang2020} and set to zero the values of the membership matrix $\hat{\bZ}$ that are smaller than $1/K$.

Figure \ref{fig:compare_methods} shows the average NVI of these methods as a function of $\rho$ under different scenarios. Most methods show an excellent performance when $\rho=0$, but as the between-community connectivity increases, the performance of all methods deteriorate. Our methods (SPCA-CD and SPCA-eig) generally achieve the best performance when the fraction of nodes in overlaps is either 0 or 10\%, and are highly competitive for 40\% in overlaps as well. OCCAM  performs well, which is reasonable since the networks were generated from this model, but it appears that in most cases we are able to fit it more accurately.  Mixed-SCORE has a good performance with no overlaps, but deteriorates quicker than other methods with the introduction of overlaps.    We should keep in mind that OCCAM and Mixed-SCORE are designed for estimating continuous memberships, and the threshold of $1/K$  to obtain binary memberships might not be optimal. While non-overlapping community detection methods can be alternatively used for the scenario when there is only a single membership per node, our methods are able to accurately assign the nodes to a single community without knowing the number of memberships.

\begin{figure}
	\centering
	\includegraphics[width=\textwidth]{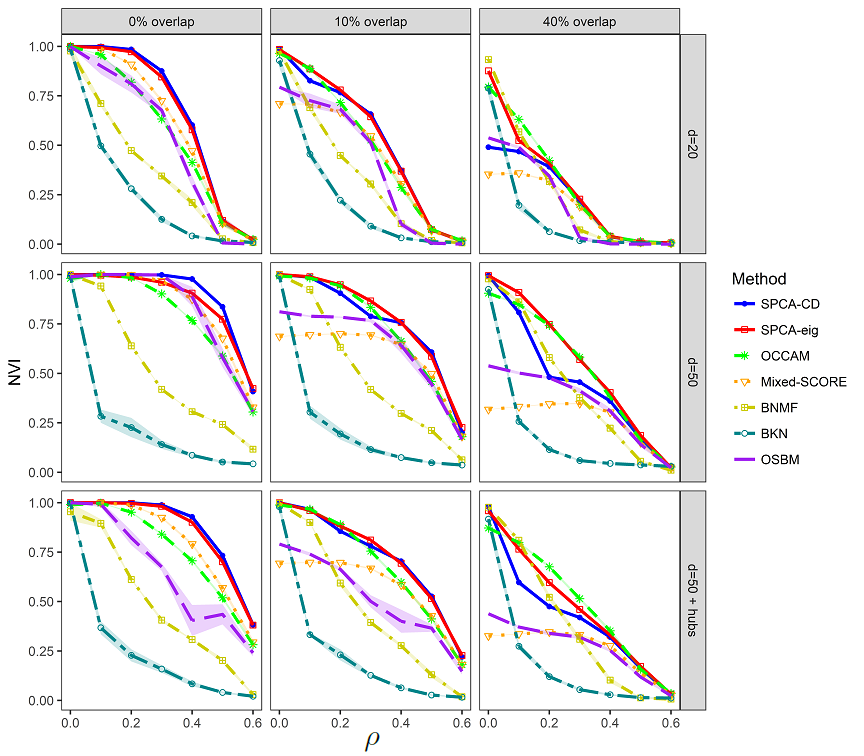}
	\caption{Performance of different methods for overlapping community detection measured by NVI, as a function of $\rho$, for different amounts of overlaps (columns) and node degrees (rows).\label{fig:compare_methods}}
\end{figure}

\subsection{Computational efficiency}
Scalability to large networks is an important issue for real data applications. Spectral methods for overlapping and non-overlapping community detection are very popular, partly due to its scalability to large networks. The accuracy of those methods usually depends on the clustering algorithm, which in practice might require multiple initial values to get an accurate result.  In contrast, our methods based on sparse principal component analysis directly estimate the membership matrix without having to estimate the eigenvectors or perform a clustering step. Although the accuracy of our methods does depend on the tuning parameter $\lambda$,   the algorithms are robust to the choice of this parameter and provide good solutions over a reasonably wide range.   

To compare computation efficiency empirically, we simulated networks with different number of communities ($K=3, 6$ and $10$) and increased the number of nodes while keeping the average degree fixed to $d=50$, with 10\% overlapping nodes. For simplicity, we used a single fixed value $\lambda=0.6$ for our methods. We initialized SPCA-CD with a random membership matrix, and  SPCA-eig with the  SPCA-CD as starting point, and therefore report the running time of SPCA-eig as the sum of the two.  We compare the performance of our methods with OCCAM, which uses a k-medians clustering to find the centroids of the overlapping communities. Since k-medians is computationally expensive and is not able to handle large networks, we also report the performance of OCCAM with the clustering step performed with k-means instead.  Additionally, we report the running time of calculating the $K$ leading eigenvectors of the adjacency matrix, which is a starting step required by  spectral methods. 
All simulations are run using Matlab R2015a.
The leading eigenvectors of the adjacency matrix are computed using the standard Matlab function \texttt{eigs}$(\cdot, K)$.

The performance in terms of time and accuracy of different methods is shown in Figure \ref{fig:computationalSPCA}. Our methods based on SPCA incur a computational cost similar to that  calculating the $K$ leading eigenvectors of the adjacency matrix, and when the number of communities is not large, our methods perform even faster. The original version of OCCAM based on k-medians clustering is limited in the size of networks it can handle, and when using k-means  the computational cost is still larger than SPCA. Our methods produce very accurate solutions in all the scenarios considered, while OCCAM deteriorates when the number of communities increases. Note that in general the performance of all methods can be improved by using different random starting values, either for clustering in OCCAM or for initializing our methods, but this will increase the computational cost; choosing tuning parameters, if the generic robust choice is not considered sufficient, will do the same.

\begin{figure}
	\centering
	\includegraphics[width=\textwidth]{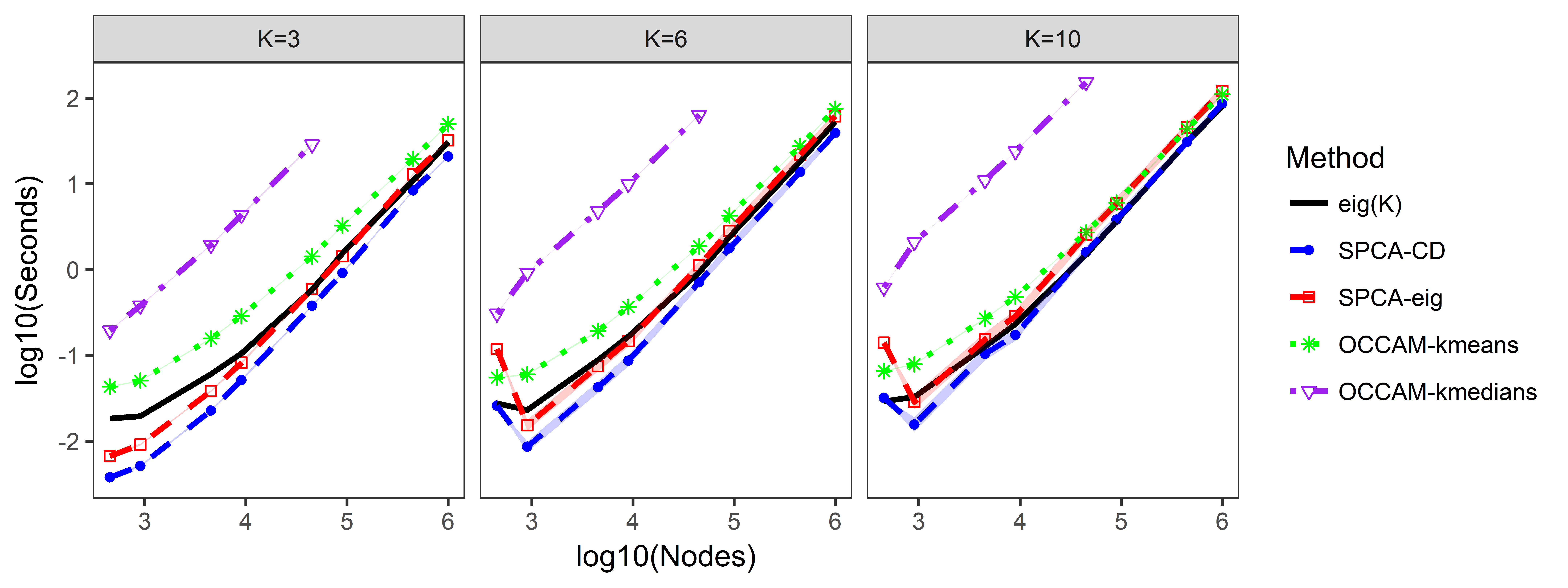}
	\includegraphics[width=\textwidth]{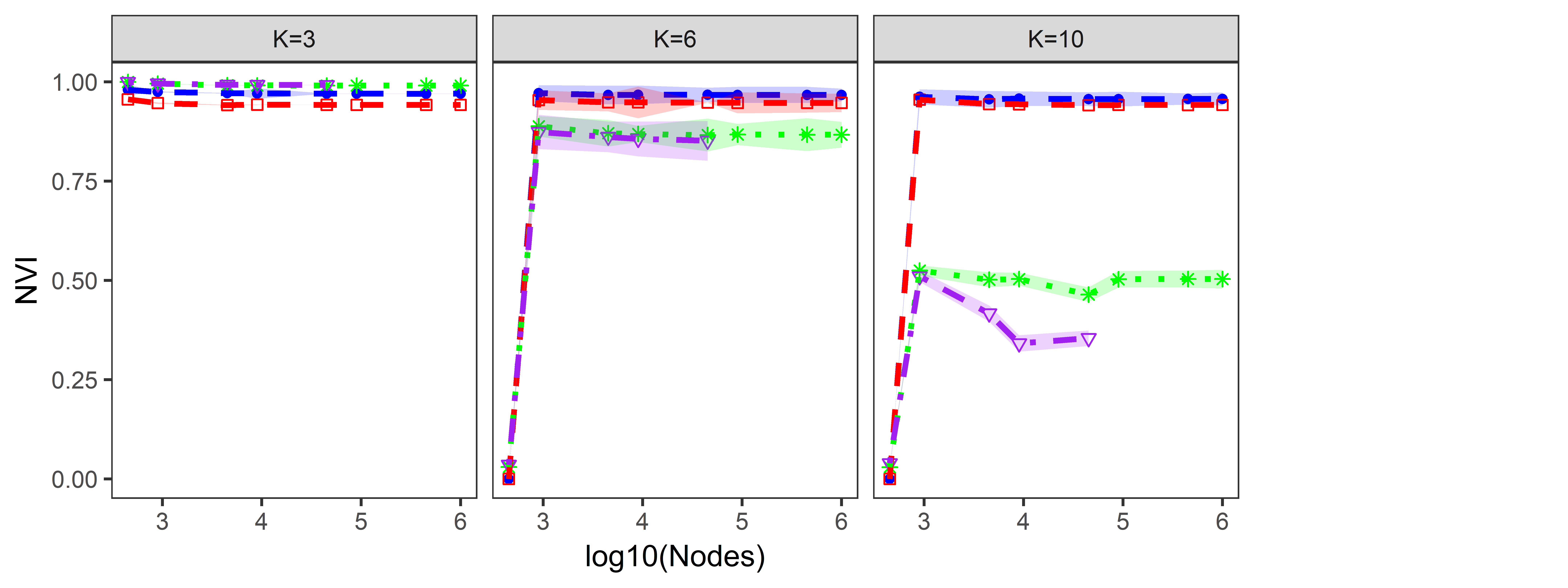}
	\caption{Performance of different methods in terms of running time (top row) and NVI (bottom row) as a function of the size of the network $n$ for varying number of communities. We compare SPCA-CD and SPCA-eig, OCCAM with two different clustering options (k-means and k-medians), and the computational cost of calculating the $K$ leading eigenvectors (eig(K)). }\label{fig:computationalSPCA}
\end{figure}

\section{Evaluation on real-world networks}

In this section, we evaluate the performance of our methods on several real-world networks. Zachary's karate club network \citep{zachary1977information} and  the political blog network \citep{adamic2005political} are two classic examples with community structure, and we start with them as an illustration. We then compare our method with other state-of-the-art overlapping community detection algorithms on the popular benchmark dataset focused specifically on overlapping communities \citep{mcauley2012learning}, which contains many social ego-networks from Facebook, Twitter and Google Plus. Finally, we use our methodology to identify communities in a novel dataset consisting of Twitter following relationship between national representatives in the Mexican Chamber of Deputies.

\subsection{Zachary's karate club network \label{sec:karate}}

\cite{zachary1977information} recorded the real-life interactions of 34 members of a karate club from a period of two years. During this period, the club split into two factions due to a conflict between the leaders, and these factions are taken to be the ground truth communities.  

We fit our methods to the karate club network, and if we use either BIC or CV to choose the optimal threshold parameter, the solution consists of two community with only pure nodes  and matches the ground truth.  This serves as reassurance that our method will not force overlaps on the communities when there are not actually there.   In contrast, OCCAM assigns 17 nodes (50\%) to both communities, and mixed-SCORE assigns 26 (76\%).

\begin{figure}
	\centering
	\includegraphics[width=0.95\textwidth]{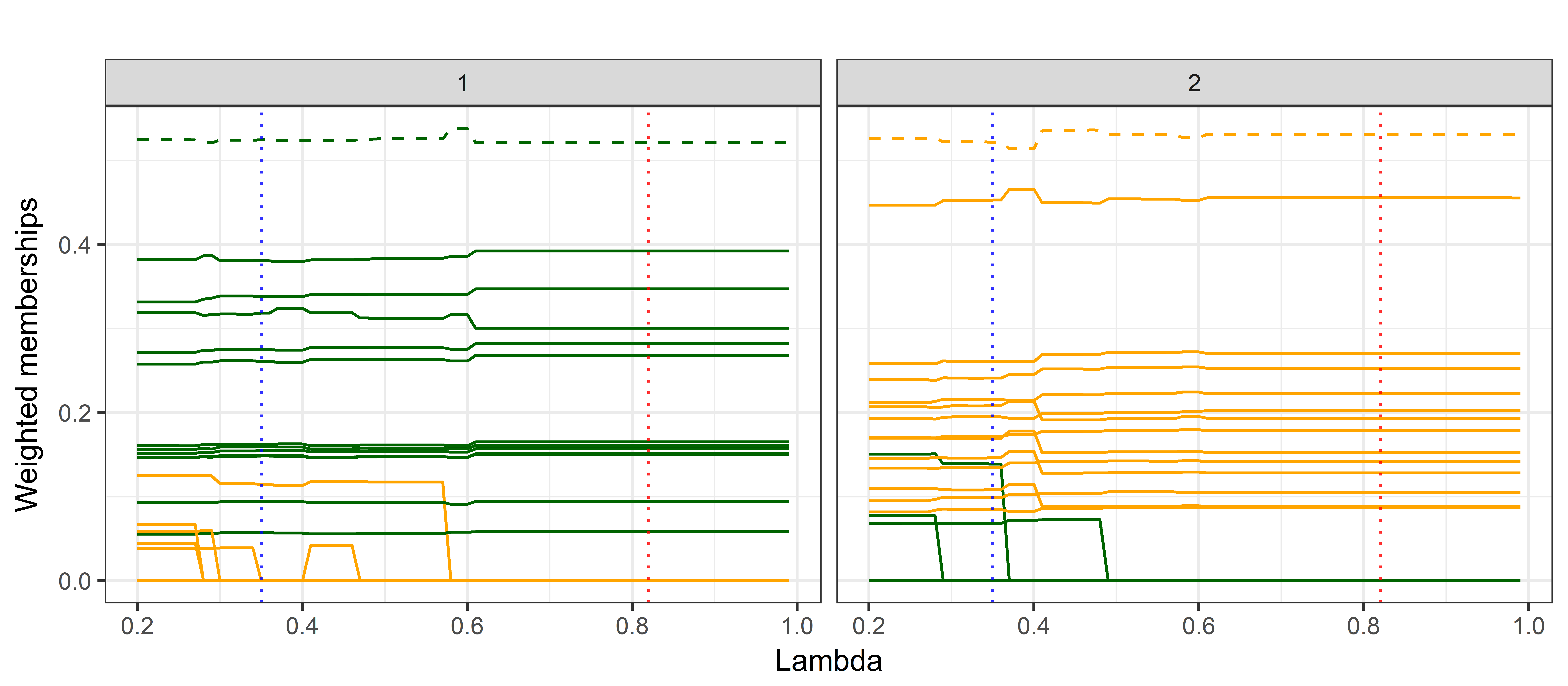}
	\caption{Node membership paths in each discovered community (left and right) as a function of the thresholding parameter $\lambda$. Colors represent ground truth, and the dashed line represents the leader of each faction. The dashed vertical lines correspond to the values of $\lambda$ chosen by BIC (blue) and CV (red). As the thresholding parameter decreases, the membership matrix becomes sparser and nodes appear on both  communities.
	}\label{fig:karatepath}
\end{figure}

If we look at the entire path over the threshold parameter $\lambda$, we can also see which nodes are potential overlaps.   Both our methods can identify community memberships, but SPCA-eig also provides information on the degree-correction parameter. In Figure \ref{fig:karatepath}, we examine the effect of the threshold parameter $\lambda$ on SPCA-eig solutions. The plots show the paths of the node membership vectors as a function of $\lambda$. Each panel corresponds to one of the columns of the membership matrix,  the colors indicate the true factions, and the paths of the faction leaders are indicated with a dashed line. The $y$-axis shows the association of the node to the corresponding community, with membership weighted by the degree-corrected parameter. In each community, the nodes with the largest values of $y$ are the faction leaders, which are connected to most of the nodes in the faction. For larger values of $\lambda$, all nodes are assigned to pure communities corresponding to true factions, but as $\lambda$ decreases the membership matrix contains more non-zero values.

\subsection{The political blogs network}
The political blog network \citep{adamic2005political} represents the hyperlinks between 1490 political blogs around the time of the 2004 US presidential election. The blogs were manually labeled as liberal or conservative, which is taken as ground truth, again without any overlaps.   This dataset is more challenging for community detection than the karate club network, due to the high degree heterogeneity \citep{karrer2011stochastic}. 
Following the literature, we focus on the largest connected component of the network, which contains 1222 nodes, and convert the edges to undirected, so $A_{ij}=1$ if either blog $i$ has an hyperlink to blog $j$ or vice versa.

Figure \ref{fig:polblog}  shows the plot of the political blog network membership paths using Algorithm \ref{alg:spca-comdet} as a function of $\lambda$ and colored by the ground truth labels. 
Using the tuning parameter selected by BIC, the algorithm assigns only 29 nodes to both communities. Other overlapping community methods assigned many more nodes to both communities:  229 (19\%) for OCCAM, and 195 (16\%) for  mixed-SCORE.  To convert the estimated solution into non-overlapping memberships in order to compare with the ground truth, each node is assigned to the community corresponding to the largest entry on the corresponding row, resulting in  52  misclustered nodes, a result similar to other community detection methods that are able to operate on networks with heterogeneous node degrees \citep{jin2015}. The membership paths that correspond to these misclustered nodes are represented with dashed lines. The fact that most of the overlapping nodes discovered by the algorithm were incorrectly  clustered supports the idea that these are indeed overlapping nodes, as the disagreement between the unsupervised clustering result and the label given by the authors might indicate that these are nodes with no clear membership.

\begin{figure}
    \centering
    \includegraphics[width=0.95\textwidth]{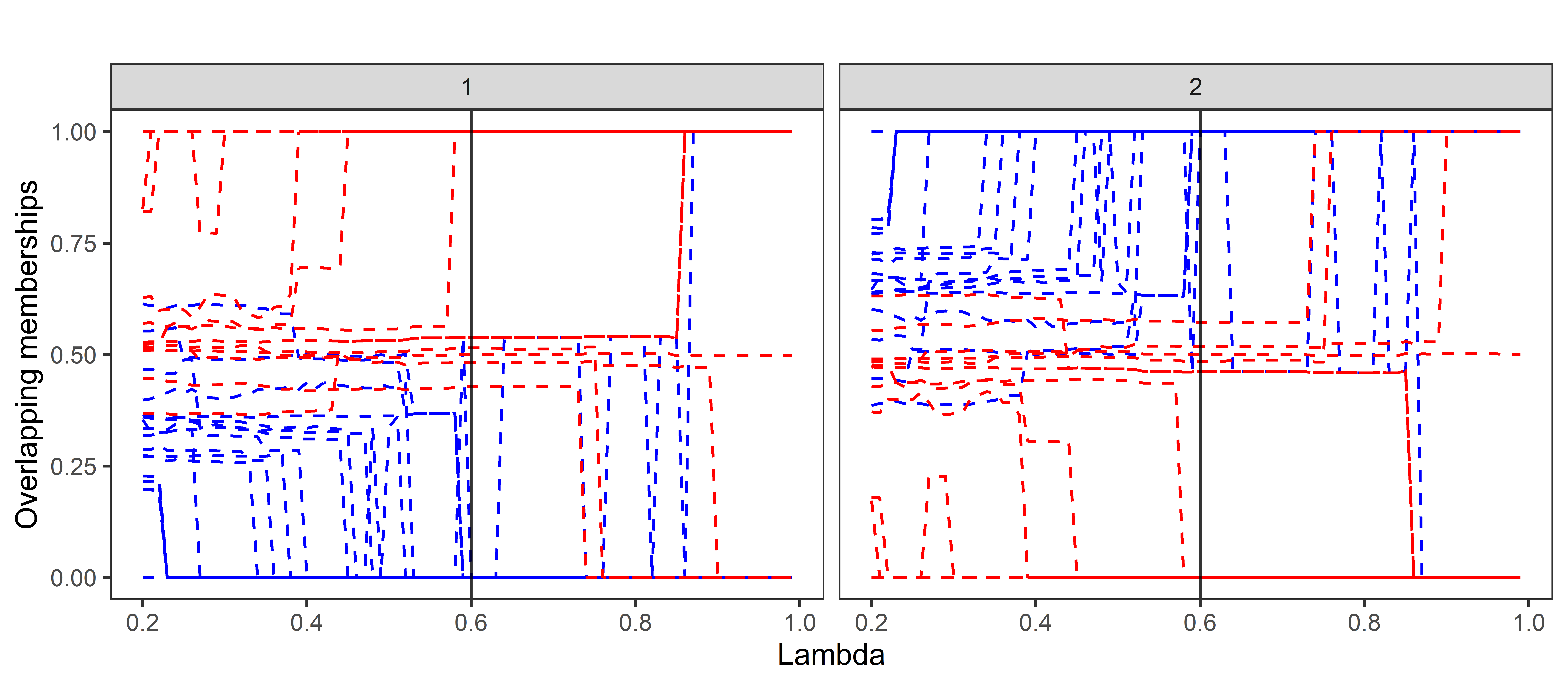}
    \caption{Node membership paths in the two discovered communities (left and right) as a function of the thresholding parameter $\lambda$. Colors represent ground truth, and the dashed lines represent  misclassified nodes. Most of the misclassified nodes are identified as overlapping.   The black vertical lines correspond to the values of $\lambda$ chosen by BIC.}
    \label{fig:polblog}
\end{figure}

\subsection{The SNAP social networks}

Social media platforms provide a rich source of data for the study of social interactions. \cite{mcauley2012learning} presented a large collection of ego-networks from Facebook, Google Plus and Twitter. An ego-network represents the virtual friendships or following-follower relationships between a group of people that are connected to a central user. Those platforms allow the users to manually label or classify their friends into groups or social circles, and this information can be used as a ground truth to compare the performance of methods for detecting communities. In \cite{Zhang2020}, several state-of-the-art overlapping community detection methods were compared on these data, showing a competitive performance of OCCAM.   We again include OCCAM  and Mixed-SCORE as examples of spectral methods for overlapping community detection. We obtained a pre-processed version of the data directly from the first author of \cite{Zhang2020};  for the details on the pre-processing steps, see Section 6 of  \cite{Zhang2020}.

\begin{table}[ht]
	\centering
	\begin{tabular}{l|rrrr}
		\hline
		Dataset (sample size) & SPCA-Eig  &  SPCA-CD    & OCCAM & M-SCORE \\ 
		\hline
		Facebook (6) & 0.573  (0.090) & $ \mathbf{0.588  (0.088)}$ & 0.548  (0.118) & 0.493  (0.137) \\ 
		Google Plus (39) & 0.408  (0.047)  & 0.427  (0.048) & $\mathbf{  0.501  (0.039)}$  & 0.475  (0.039) \\ 
		Twitter (168) & 0.435  (0.021) & $\mathbf{ 0.477  (0.021) }$ &  0.450  (0.021) & 0.391  (0.020) \\ 
		\hline
	\end{tabular}
	\caption{Average performance measured by NVI (with standard errors in parentheses) of overlapping community detection methods on SNAP ego-networks.}\label{tab:egonet}
\end{table}

Table \ref{tab:egonet} shows the average performance measured by NVI for the community detection methods we compared. For our methods, the value of $\lambda$  was chosen by BIC, like in simulations. For OCCAM and mixed-SCORE, we thresholded  continuous membership assignments at $1/K$. Our methods (SPCA-eig and SPCA-CD) show a slightly better performance than the rest of the methods in the Facebook networks. SPCA-CD performs better than other methods on the Twitter networks, but SPCA-eig does not perform better than OCCAM. For Google Plus networks, OCCAM and mixed-SCORE have a clear advantadge. Figure \ref{fig:egonet-hist} presents a visualization of the distribution of of several network summary statistics  for each social media platform.    It suggests that Google Plus networks might be harder because they tend to have more overlaps between communities, although they also tend to have more nodes.   Facebook networks, in contrast, have higher modularity values and smaller overlaps, and thus should be easier to cluster. In general, all methods perform reasonably, with SPCA-CD given the best overall performance Facebook and Twitter networks, and OCCAM being overall best on Google Plus.  This is consistent with what we observed in simulations and what we would expect by design:  our methods are more likely to perform better than others when membership vectors are sparse.  

\begin{figure}
	\includegraphics[width = \textwidth]{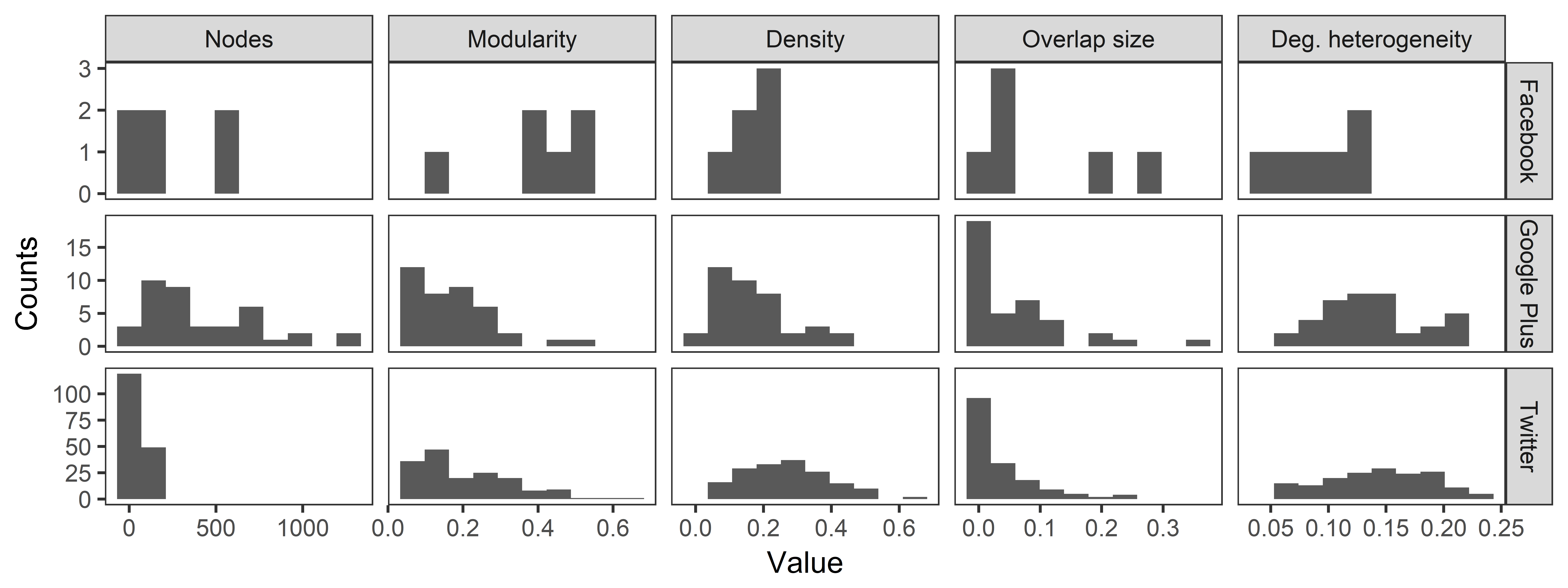}
	\caption[Histograms of summary statistics for SNAP ego-networks.]{Histograms of summary statistics for SNAP ego-networks. \citep{mcauley2012learning}. The histograms show the number of nodes ($n$), Newman-Girvan modularity, density (number of edges divided by $n^2$), overlap size (percentage of nodes with overlapping memberships) and degree heterogeneity measured by the standard deviation of node degrees divided by $n$).}\label{fig:egonet-hist}
\end{figure}

\subsection{Twitter network of Mexican Representatives}

We consider the Twitter network between members of the  Mexican Chamber of Deputies (the lower house of the Mexican parliament), from  the LXIII Legislature for the period of 2015-2018. The network captures a snapshot of Twitter data from December 31st, 2017, and has 409 nodes corresponding to the representatives with a valid Twitter handle.  Two nodes are connected by an edge if at least one of them follows the other on Twitter;  we ignore the direction.   Each member belongs to one of eight political parties or is an independent, resulting in $K=9$ true communities;  see Figure~\ref{fig:Mexican}. The data can be downloaded from \url{https://github.com/jesusdaniel/spcaCD/data}.

\begin{figure}
    \centering
    \includegraphics[width=0.48\textwidth]{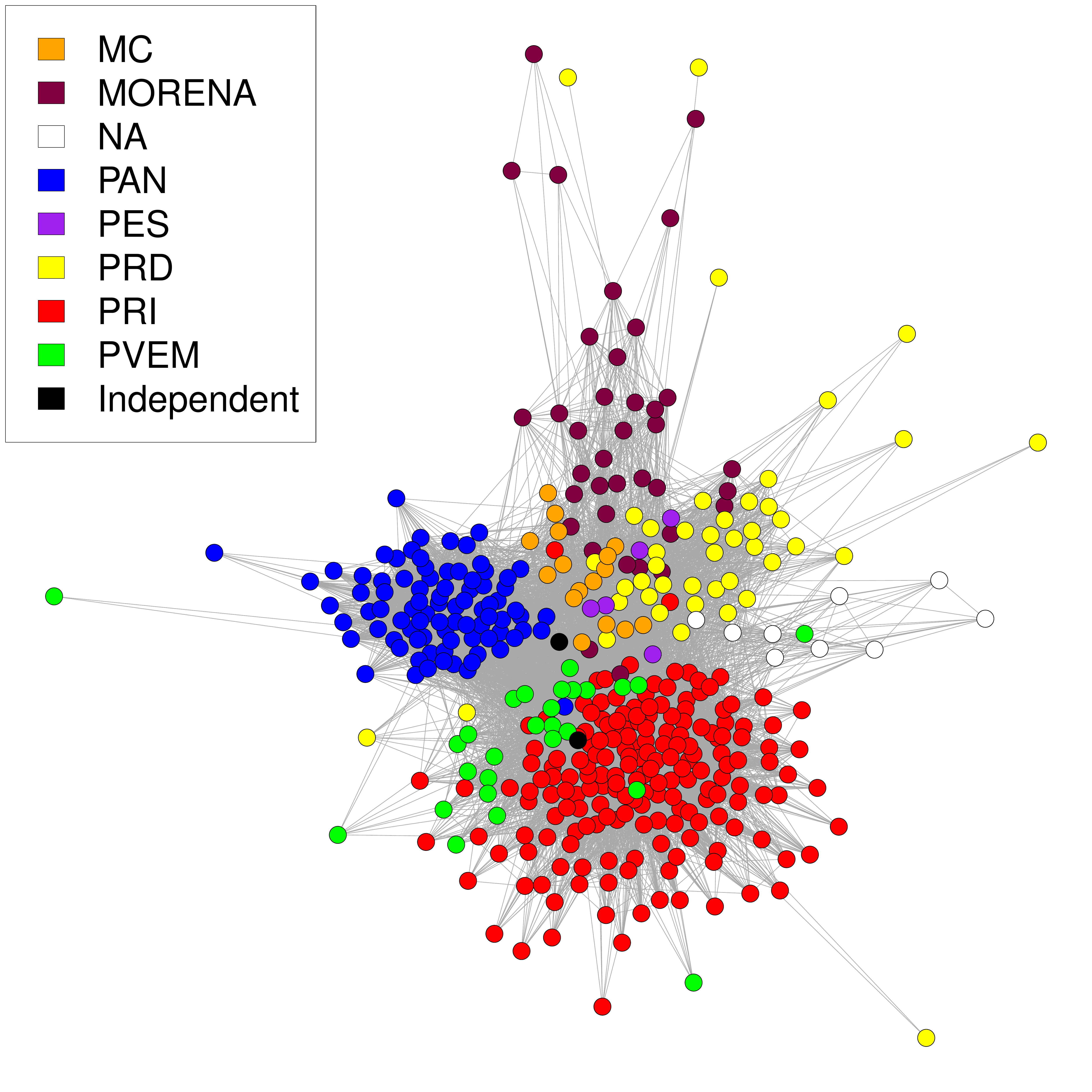}
    \includegraphics[width=0.48\textwidth]{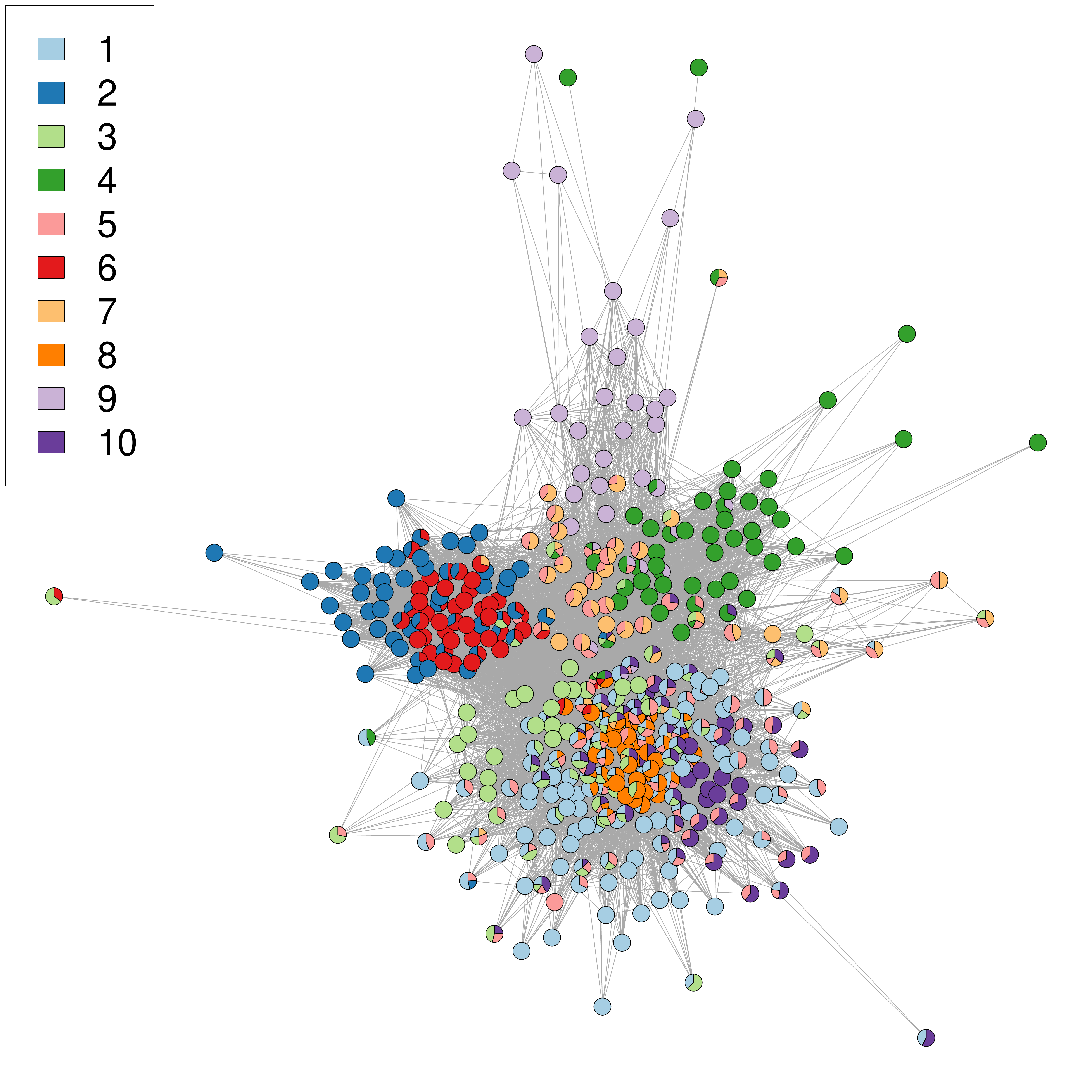}
    \caption{Twitter network of Mexican Chamber of Deputies from 2017.  Left: party affiliations (true communities);  right:  output of Algorithm 1 with estimated $K=10$.}
    \label{fig:Mexican}
\end{figure}

We apply Algorithm~\ref{alg:spca-eig} to this network, using 20-fold edge cross-validation to estimate the number of communities and choose thresholding parameters.  Figure~\ref{fig:CV-MexicanRepresentatives} shows the average MSE across all the folds,  minimized when $K=10$. However, solutions corresponding to all $K$ from 8 to 11 are qualitatively very similar, the only difference being that the largest parties (PRI and PAN) get split into smaller communities, with clusters containing  the most popular  members of each party, and/or factions within a party that are more connected  to some of the other parties.

\begin{figure}
    \centering
    \includegraphics[width=0.5\textwidth]{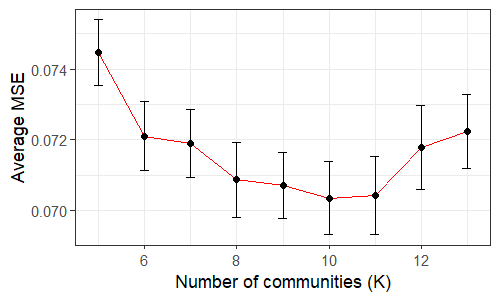}
    \caption{Average cross-validation mean squared error (MSE) over 20 cross-validation folds as a function of the number of communities (the error bars indicate two standard errors). 
    }
    \label{fig:CV-MexicanRepresentatives}
\end{figure}

A comparison between the estimated  membership vectors and party affiliations reveals that our algorithm discovers meaningful overlapping communities. Table~\ref{table:Memberships-MexicanRep} compares the estimated  overlapping memberships with the party labels, by counting the number of nodes that are assigned to a given community  (recall that each node can be assigned to more than one community) and belong to a specific party.  Some of the communities contain representatives from two or more different parties, which is a reflection of coalitions and factions. For example, the majority of nodes in community 3 belong to either PRI or PVEM, which formed a coalition during the preceding election  in 2015. On the other hand, nodes from MORENA in community 4 were  members of PRD before MORENA  was formed in 2014. The plot on the right in Figure~\ref{fig:Mexican} also shows a significant overlap between these parties. 

Exploring  individual memberships reveals that the number of communities a node is assigned to seems associated with its overall popularity in the network. For example, the 
node with the largest number of community memberships (7 in total) is the representative with the largest degree in the network, while the second largest number of memberships (5 in total) is the president of the Chamber of Deputies in 2016.

\begin{table}[ht]
\centering
\begin{tabular}{l|rrrrrrrrrr}
  \hline
  &  \multicolumn{10}{c}{Estimated communities}\\
Party (\#nodes) & {\bfseries 1} & {\bfseries 2} & {\bfseries 3} & {\bfseries 4} & {\bfseries 5} & {\bfseries 6} & {\bfseries 7} & {\bfseries 8} & {\bfseries 9} & {\bfseries 10} \\ 
  \hline
 MC (16) &   0 &   0 &   0 &   0 &  14 &   0 &  16 &   0 &   0 &   0 \\ 
  MORENA (37) &   1 &   0 &   1 &   9 &   4 &   0 &   2 &   0 &  34 &   1 \\ 
  NA (9) &   2 &   0 &   4 &   1 &   8 &   0 &   9 &   0 &   0 &   1 \\ 
   PAN (84) &   0 &  53 &   2 &   0 &   2 &  49 &   3 &   1 &   0 &   0 \\ 
   PES (5) &   0 &   0 &   2 &   0 &   3 &   0 &   5 &   0 &   0 &   1 \\ 
   PRD (44) &   2 &   1 &   3 &  42 &   2 &   0 &   2 &   0 &   0 &   3 \\ 
   PRI (185) & 133 &   1 &  52 &   3 &  71 &   3 &  12 &  53 &   2 &  59 \\ 
   PVEM (27) &   1 &   0 &  27 &   0 &   3 &   1 &   0 &   1 &   0 &   0 \\ 
  Independent (2) &   1 &   0 &   0 &   0 &   1 &   0 &   1 &   1 &   0 &   0 \\ 
   \hline
\end{tabular}
\caption{Each entry indicates the number of nodes from the party (row) assigned to a given community (column).  Nodes are counted multiple times if they were assigned to more than one community.}
\label{table:Memberships-MexicanRep}
\end{table}

%% file: discussion.tex
\section{Discussion}
We presented an approach to estimate a regularized basis of the principal subspace of the network adjacency matrix, and showed that its sparsity pattern encodes the community membership information.    Varying the amount of regularization controls the sparsity of the node memberships and allows to one to obtain a family of solutions of increasing complexity.   These methods show good accuracy in estimating the  memberships, and are computationally very efficient allowing to scale well to large networks. Our present theoretical results are limited to fixed points of the algorithms; establishing theoretical guarantees in more general settings as analyzing conditions for convergence to the fixed point are  left for future work. 

Spectral inference has been used for multiple tasks on networks:  community detection \citep{lei2015consistency,le2017concentration}, hypothesis testing \citep{tang2017semiparametric}, multiple network dimensionality reduction \citep{levin2017central} and network classification \cite{arroyo2020simultaneous}. While the principal eigenspace of the adjacency matrix can provide the information needed for these problems, our results suggest that regularizing  the eigenvectors can lead to improved estimation and computation in community detection;  exploring the effects of this type of regularization in other network tasks is a promising direction for future work.

%% file: appendix2.tex
\appendix
\section{Appendix}
\begin{proof}[Proof of Proposition \ref{prop:identifiabilitySPCA}]
	Because $\bV$ and $\widetilde{\bV}$ are two bases of the column space of $\bP$, and $\operatorname{rank}(\bP)=K$, then $\bP=\bV\bU^\top=\widetilde{\bV}\widetilde{\bU}^\top$ for some full rank matrices $\bU,\widetilde{\bU}\in\Bbb{R}^{n\times K}$ and therefore 
	\begin{equation}
	\bV=\widetilde{\bV}(\widetilde{\bU}^\top\bU)({\bU}^\top\bU)^{-1}. \label{eq:identifiability1}
	\end{equation}
	Let $(\widetilde{\bU}^\top\bU)({\bU}^\top\bU)^{-1}=\bLambda$.
	We will show that $\bLambda=\bQ\bD$ for a permutation matrix $\bQ\in\{0,1\}^{K\times K}$ and a diagonal matrix $\bD\in\real^{K\times K}$, or in other words, this is a generalized permutation matrix.
	
	Let $\btheta,\widetilde{\btheta}\in\real^n$ and $\bZ,\widetilde{\bZ}\in\real^{n\times K}$ such that $\btheta_i = \left(\sum_{k=1}^K\bV_{ik}^2\right)^{1/2}$, $\widetilde{\btheta}_i = \left(\sum_{k=1}^K\widetilde{\bV}_{ik}^2\right)^{1/2}$, and $\bZ_{ik}= \bV_{ik}/\btheta_i$ if $\btheta_i>0$, and $\bZ_{ik}=0$ otherwise (similarly for $\widetilde{\bZ}$).
	Denote by $\mathcal{S}_1=(i_1,\ldots,i_K)$ to the vector of row indexes that satisfy $\bV_{i_jj}> 0$ and $\bV_{i_jj'}=0$ for $j'\neq j$, and $j=1,\ldots,j$ (these indexes exist by  assumption). In the same way, define $\mathcal{S}_2=(i'_1,\ldots,i'_K)$ such that $\widetilde{\bV}_{i'_jj}> 0$ and $\widetilde{\bV}_{i_jj'}=0$ for $j'\neq j$. $j=1,\ldots,j$. Denote by $\bZ_{\mathcal{S}}$ to the $K\times K$ matrix  formed by the rows indexed by $\mathcal{S}$. Therefore 
	$$\bZ_{\mathcal{S}_1}= \bI_K=\widetilde{\bZ}_{\mathcal{S}_2}.$$
	Write $\bTheta = \diag(\btheta)\in\real^{n\times n}$ and $\widetilde\bTheta = \diag(\widetilde\btheta)\in\real^{n\times n}$. From Equation~\eqref{eq:identifiability1} we have 
	\begin{align*}
	(\bTheta\bZ)_{\mathcal{S}_2}= & (\widetilde{\bTheta}\widetilde{\bZ})_{\mathcal{S}_2}\bLambda   =\widetilde{\bTheta}_{\mathcal{S}_2, \mathcal{S}_2} \widetilde{\bZ}_{\mathcal{S}_2}\bLambda =  \widetilde{\bTheta}_{\mathcal{S}_2, \mathcal{S}_2}\bLambda\ ,
	\end{align*}
	where $\Theta_{\mathcal{S}, \mathcal{S}}$ is the submatrix of $\Theta$ formed by the rows and columns indexed by $\mathcal{S}$.
	Thus, 
	$$\bLambda = (\widetilde{\bTheta}_{\mathcal{S}_2, \mathcal{S}_2}^{-1}{\bTheta}_{\mathcal{S}_2, \mathcal{S}_2})\bZ_{\mathcal{S}_2},$$
	which implies that $\bLambda$ is a non-negative matrix. Applying the same to the equation  $(\bTheta\bZ)_{\mathcal{S}_1}\bLambda^{-1}=  (\widetilde{\bTheta}\widetilde{\bZ})_{\mathcal{S}_1}$, we have 
	$$\bLambda^{-1} = ({\bTheta}_{\mathcal{S}_1, \mathcal{S}_1}^{-1}\widetilde{\bTheta}_{\mathcal{S}_1, \mathcal{S}_1})\widetilde{\bZ}_{\mathcal{S}_1}.$$
	Hence, both $\bLambda$ and $\bLambda^{-1}$ are non-negative matrices, which implies that $\bLambda$ is a positive generalized permutation matrix, so $\bLambda=\bQ\bD$ for some permutation matrix $\bQ$ and a diagonal matrix $\bD$ with $\diag(\bD)>0$.
\end{proof}

\begin{proof}[Proof of Proposition~\ref{prop:V-occam}]
	Let $\btheta\in\real^n$ be a vector such that $\btheta_i^2=\sum_{k=1}^K\bV_{ik}^2$, and define $\bZ\in\real^{n\times K}$ such that $\bZ_{ik}=\frac{1}{\theta_i}\bV_{ik}$, for each $i\in[n], k\in[K]$. Let $\bB = (\bV^\top\bV)^{-1}\bV^T\bU$.	 To show that $\bB$ is symmetric, observe that $\bV\bU^\top = \bP = \bP^\top = \bU\bV^\top$.   Multiplying both sides by $\bV$ and $\bV^\top$,
	$$\bV^\top \bV \bU^\top \bV = \bV^\top \bU\bV^\top \bV,$$
	and observing that $(\bV^\top\bV)^{-1}$ exists since $\bV$ is full rank, we have 
	$$ \bU^\top \bV(\bV^\top \bV)^{-1} = (\bV^\top \bV)^{-1}\bV^\top \bU,$$
	which implies that $\bB^\top = \bB$. 
	To obtain the equivalent representation for $\bP$, form a diagonal matrix $\bTheta=\diag(\btheta)$. Then $\bTheta\bZ= \bV$, and
	\begin{align*}
	\bTheta\bZ\bB\bZ^\top \bTheta & = \bV[(\bV^T\bV)^{-1}\bV^T\bU]\bV^{\top} 
	 =  \bV(\bV^T\bV)^{-1}\bV^T\bV\bU^{\top} 
	 = \bV\bU^\top = \bP.
	\end{align*}
	Finally, under the conditions of Proposition~\ref{prop:identifiabilitySPCA}, $\bV$ uniquely determines the pattern of zeros of any non-negative eigenbasis of $\bP$, and therefore $\supp(\bV) = \supp(\bTheta\bZ\bQ) = \supp(\bZ\bQ)$ for some permutation $\bQ$.
\end{proof}

\begin{proof}[Proof of Proposition \ref{prop:fixed-point1}]
Suppose that $\bP=\bV\bU^\top$ for some non-negative matrix $\bV$ that satisfies the assumptions of Proposition~\ref{prop:identifiabilitySPCA}. Let $D\in\real^K$ such that $D_i= {\|\bV_{\cdot k}\|_2}$ and $\bD=\diag(D)$. Then $\bP = \widetilde{\bV}\bD\bU^\top$. Let
$\bV^{(0)}  = \widetilde{\bV}$ be the initial value of Algorithm~\ref{alg:spca-eig}. Then, observe that
$$\bT^{(1)} = \bP\widetilde\bV = \widetilde\bV\bD\bU^T\widetilde\bV,$$
\begin{align*}
    \widetilde{\bT}^{(1)} & = \bT^{(1)}\left[\widetilde\bV^\top\bT^{(1)}\right]^{-1}(\widetilde\bV^T\widetilde\bV)\\
    & = \widetilde\bV\bD (\bU^T\widetilde\bV)\left(\bU^T\widetilde\bV\right)^{-1}\bD^{-1}(\widetilde\bV^\top\widetilde\bV)^{-1}(\widetilde\bV^\top\widetilde\bV)\\
    & = \widetilde\bV.
\end{align*}
Suppose that $\lambda\in[0,v^\ast).$ Then,   $\lambda\max_{j\in[K]}|\widetilde{\bV}| <\widetilde\bV_{ik}$ for all $i\in[n],k\in[K]$ such that $\bV_{ik}>0$, and hence $\bU^{(1)}=\mathcal{S}(\widetilde\bV, \lambda) = \widetilde\bV$. Finally, since $\|\widetilde\bV_{\cdot,k}\|_2=1$ for all $k\in[K]$, then $\bV^{(1)}=\widetilde\bV$. 
\end{proof}

\begin{proof}[Proof of Theorem \ref{theo-sparsistency}]
	The proof consists of a one-step fixed point
	analysis of Algorithm \ref{alg:spca-comdet}. We will show that if $\bZ^{(t)}= \bZ$, then $\bZ^{(t+1)}=\bZ$ with high probability. Let $\bT=\bT^{(t+1)}=\bA\bZ$ be value after the multiplication step. Define $\bC\in\Bbb{R}^{K\times K}$ to be the diagonal matrix with community sizes on the diagonal, $\bC_{kk} = n_k=\|\bZ_{\cdot, k}\|_1$. Then $\widetilde{\bT}=\widetilde{\bT}^{(t+1)}= \bT\bC^{-1}$. In order for the threshold to set the correct set of entries to zero, a sufficient condition is that in each row $i$ the largest element of $\widetilde{\bT}_{i,\cdot}$ corresponds to the correct community. Define $\mathcal{C}_k\subset [n]$ as the node subset corresponding to community $k$. Then,
	\begin{equation*}
	\widetilde{\bT}_{ik} = \frac{1}{{n_k}} \bA_{i,\cdot} \bZ_{\cdot, k} = \frac{1}{{n_k}}\sum_{j\in\mathcal{C}_k}\bA_{ij}.
	\end{equation*}
	Therefore $\widetilde{\bT}_{ik}$ is a sum of independent and identically distributed Bernoulli random variables. Moreover, for each $k_1$ and $k_2$ in $[K]$, $\widetilde{\bT}_{ik_1}$ and $\widetilde{\bT}_{ik_2}$ are independent of each other. 
	
	Given a value of $\lambda\in(0,1)$, let 
	$$\mathcal{E}_i(\lambda)=\{ \lambda |\widetilde{\bT}_{ik_i}|>  |\widetilde{\bT}_{ik_j}|, i\in\mathcal{C}_{k_i} \forall k_j\neq k_i\}$$ 
	be the event that the largest entry of $\widetilde{\bT}_{i\cdot}$ corresponds to $k_i$, that is, the entry corresponding to the community of node $i$, and all the other indexes in that row are smaller in magnitude than $\lambda |\widetilde{\bT}_{ik_i}|$.   Let $\bU = \bU^{(t+1)}=\mathcal{S}(\widetilde{\bT}^{(t+1)}, \lambda)$ be the matrix obtained after the thresholding step. Under the event $\mathcal{E}(\lambda)=\bigcap_{i=1}^n \mathcal{E}_i(\lambda)$, we have that $\|\bU_{i,\cdot}\|_\infty = \bU_{ik_i}$ for each $i\in[n]$, and hence
	\[\bU_{ik}= \left\{\begin{array}{cl}
	\bU_{ik_i} & \text{if }k=k_i,\\
	0 & \text{otherwise.}
	\end{array}\right. \]
	Therefore, under the event $\mathcal{E}(\lambda)$, the thresholding step recovers the correct support, so $\bZ^{(t+1)}=\bZ$. 
	
	Now we verify that under the conditions of Theorem~\ref{theo-condseparation}, the event $\mathcal{E}(\lambda)$ happens with high probability. By a union bound,
	\begin{align}
	    	\Bbb{P}(\mathcal{E}(\lambda))  \geq  1-\sum_{i=1}^n \Bbb{P}(\mathcal{E}_i(\lambda)^C) 
	 \geq  1-\sum_{i=1}^n\sum_{j\neq k_i}\Bbb{P}(\widetilde{\bT}_{ij}>\lambda \widetilde{\bT}_{ik_i}).\label{eq:theosparsistency-unionbound}
	\end{align}
For $j\neq k_i$, $\widetilde{\bT}_{ij}-\lambda \widetilde{\bT}_{ik_i}$ is a sum of independent random variables with expectation
	\begin{eqnarray}
		\Bbb{E}\left[\widetilde{\bT}_{ij}-\lambda \widetilde{\bT}_{ik_i}\right]  =  \frac{1}{n_j}\sum_{j\in\mathcal{C}_j}\Bbb{E}[\bA_{ij}] - 
		\frac{\lambda}{n_{k_i}}\sum_{j\in\mathcal{C}_{k_i}}\Bbb{E}[\bA_{ij}]\nonumber\\
		 =  q - \lambda \frac{n_{k_i}-1}{n_{k_i}}p.\label{eq:proofE}
	\end{eqnarray}
	By Hoeffding's inequality, we have that for any $\tau \in\Bbb{R}$,
	\begin{align*}
		\Bbb{P}\left(\widehat{\bT}_{ij}-\lambda \widehat{\bT}_{ik_i}  \geq  \tau + \Bbb{E}\left[\widehat{\bT}_{ij} - \lambda \widehat{\bT}_{ik_i}\right] \right) 
		& \leq 2\exp\left(\frac{-2\tau^2}{\frac{1}{n_j} + \frac{\lambda^2}{n_{k_i}}}\right)\\
		& \leq   2\exp\left(- \frac{2n_{\min} \tau^2}{1+\lambda^2}\right)
		 \leq   2\exp\left(-n_{\min} \tau^2\right),
	\end{align*}
	where $n_{\min} = \min_{k\in[K]}n_k$.
	Setting
	\begin{align*}
	\tau = -\Bbb{E}\left[\widehat{\bT}_{ij}- \lambda \widehat{\bT}_{ik_i}\right] 
	\geq  \lambda^\ast p - q - \frac{1}{n_{k_i}}p   
	\end{align*}
	and using \eqref{eq:proofE} and \eqref{theo-condseparation}, we obtain that for $n$ sufficiently large,
	\begin{align*}
	\Bbb{P}\left(\widehat{\bT}_{ij}>\lambda \widehat{\bT}_{ik_i}  \right) &  \leq 2\exp\left(-n_{\min} \left(c_1\sqrt{\frac{\log(Kn)}{\min_{k}n_k}} -\frac{p}{n_k} \right)^2  \right)\\
	& \leq 2\exp\left(-n_{\min} \left((c_1-1)\frac{\log(Kn)}{n_{\min}} \right)  \right)= \frac{2}{(Kn)^{c_1-1}}
	\end{align*}
	Combining with the bound \eqref{eq:theosparsistency-unionbound},  the probability of event $\mathcal{E}(\lambda)$  (which implies that $\bZ^{(t+1)} = \bZ$) is bounded from below as  
	\begin{align*}
	\Bbb{P}(\mathcal{E}(\lambda)) & \geq 1- n(K-1)\min_{i\in[n], k\in[K]}\Bbb{P}\left(\widehat{\bT}_{ij}>\lambda \widehat{\bT}_{ik_i}  \right)\\
	& \geq 1-\frac{2(K-1)n}{(Kn)^{c_1-1}}\\
	&\geq 1- \frac{2}{Kn^{(c_1 -2)}}.
	\end{align*}
Therefore, with high probability $\bZ$ is a fixed point of the Algorithm\ref{alg:spca-comdet} for any $\lambda\in(\lambda^\ast,1)$.
\end{proof}

\begin{proof}[Proof of Proposition~\ref{prop:estimator}]
Observe that
\begin{align*}
    \|\bA -\widehat{\bV}\bB\widehat{\bV}^\top\|_F^2 & = \text{Tr}(\bA^\top\bA) - 2\text{Tr}(\widehat{\bV}^\top\bA^\top\widehat{\bV}\bB) + \text{Tr}(\bB^\top\widehat{\bV}^\top\widehat{\bV}\bB\widehat{\bV}^\top\widehat{\bV})\\
    & = \|\bB - (\widehat{\bV}^\top\widehat{\bV})^{-1}\widehat{\bV}^\top\bA\widehat{\bV}(\widehat{\bV}^\top\widehat{\bV})^{-1}\|_F^2 + C,
\end{align*}
where $c$ is a constant that does not depend on $\bB$. Therefore $\widehat{\bB}$
\begin{align*}
    \widehat{\bP} & = \argmin_{\bB\in\real^{K\times K}}\|\bA -\widehat{\bV}\bB\widehat{\bV}^\top\|_F^2\\
    & = \widehat{\bV}(\widehat{\bV}^\top\widehat{\bV})^{-1}\widehat{\bV}^\top\bA\widehat{\bV}(\widehat{\bV}^\top\widehat{\bV})^{-1}\widehat{\bV}^\top.
\end{align*}
Suppose that $\widehat{\bV} = \widehat{\bQ}\widehat{\bR}$ for some matrix $\bQ$ with orthonormal columns of size $n\times K$. Then, $\widehat{\bR}$ is a full rank matrix, and therefore
$$(\widehat{\bV}^\top\widehat{\bV})^{-1} = \widehat{\bR}^{-1}(\widehat{\bQ}^\top\widehat{\bQ})^{-1}(\widehat{\bR^\top})^{-1}.$$
Using this equation, we obtain the desired result.
\end{proof}